\documentclass[10pt,journal,compsoc]{IEEEtran}

\ifCLASSINFOpdf
\else
\fi

\usepackage{amsmath,amssymb}
\usepackage{amsthm}
\usepackage{algorithmic}
\usepackage{array}
\usepackage{longtable}

\ifCLASSOPTIONcaptionsoff
 \usepackage[nomarkers]{endfloat}
\let\MYoriglatexcaption\caption
 \renewcommand{\caption}[2][\relax]{\MYoriglatexcaption[#2]{#2}}
\fi

\usepackage{url}

\usepackage{graphicx}
\usepackage{subfigure}
\usepackage{multirow}
\usepackage{accents}
\usepackage{ragged2e}
\usepackage[colorlinks]{hyperref}
\usepackage{epstopdf}
\usepackage{nonfloat}
\usepackage{cite}
\usepackage{epstopdf}

\newcommand{\ee}{\end{equation}}
\newcommand{\bea}{\begin{eqnarray}}
\newcommand{\eea}{\end{eqnarray}}
\newcommand{\bean}{\begin{eqnarray*}}
\newcommand{\eean}{\end{eqnarray*}}
\newcommand{\ba}{\begin{array}}
\newcommand{\ea}{\begin{array}}

\begin{document}
%
\title{A Scalable Optimization Mechanism for Pairwise based Discrete Hashing}

\author{Xiaoshuang~Shi,
         Fuyong~Xing,
         Zizhao~Zhang,
         Manish~Sapkota,
         Zhenhua ~Guo,
         and Lin~Yang
\thanks{X. Shi and L. Yang are with the J. Crayton Pruitt Family Department of Biomedical Engineering, University of Florida, Gainesville, FL, USA, e-mail: (xsshi2015@ufl.edu, lin.yang@bme.ufl.edu).}
\thanks{F. Xing and M. Sapkota is with the Department of Electrical and Computer Engineering, University of Florida, Gainesville, FL, USA, e-mail: (f. xing@ufl.edu and manish.sapkota@gmail.com).}
\thanks{Z. Zhang is with the Department of Computer Science and Engineering, University of Florida, Gainesville, FL, USA, e-mail: (mr.zizhaozhang@gmail.com).}
\thanks{Z. Guo is with Graduate School at Shenzhen, Tsinghua University, Shenzhen, Guodong, China, e-mail: (zhenhua.guo@sz.tsinghua.edu.cn).}
}

\IEEEtitleabstractindextext{
\begin{abstract}
\justifying
 Maintaining the pair similarity relationship among originally high-dimensional data into a low-dimensional binary space is a popular strategy to learn binary codes. One simiple and intutive method is to utilize two identical code matrices produced by hash functions to approximate a pairwise real label matrix. However, the resulting quartic problem is difficult to directly solve due to the non-convex and non-smooth nature of the objective. In this paper, unlike previous optimization methods using various relaxation strategies, we aim to directly solve the original quartic problem using a novel alternative optimization mechanism to linearize the quartic problem by introducing a linear regression model. Additionally, we find that gradually learning each batch of binary codes in a sequential mode, i.e. batch by batch, is greatly beneficial to the convergence of binary code learning. Based on this significant discovery and the proposed strategy, we introduce a scalable symmetric discrete hashing algorithm that gradually and smoothly updates each batch of binary codes. To further improve the smoothness, we also propose a greedy symmetric discrete hashing algorithm to update each bit of batch binary codes. Moreover, we extend the proposed optimization mechanism to solve the non-convex optimization problems for binary code learning in many other pairwise based hashing algorithms. Extensive experiments on benchmark single-label and multi-label databases demonstrate the superior performance of the proposed mechanism over recent state-of-the-art methods. 
\vspace{-0.5em}
\end{abstract}}

\maketitle
\IEEEdisplaynontitleabstractindextext
\IEEEpeerreviewmaketitle

\section{Introduction}
\label{introduction}
Hashing has become a popular tool to tackle large-scale tasks in information retrieval, computer vision and machine leaning communities, since it aims to encode originally high-dimensional data into a variety of compact binary codes with maintaining the similarity between neighbors, leading to significant gains in both computation and storage \cite{ARY} \cite{WMSS} \cite{cao2018binary}. 

Early endeavors in hashing focus on data-independent algorithms, like locality sensitive hashing (LSH) \cite{PR} and min-wise hashing (MinHash) \cite{broder2000min} \cite{shrivastava2014defense}. They construct hash functions by using random projections or permutations. However, due to randomized hashing, in practice they usually require long bits to achieve high precision per hash table and multiple tables to boost the recall \cite{WMSS}. To learn compact binary codes, data-dependent algorithms using available training data to learn hash functions have attracted increasing attention. Based on whether utilizing semantic label information, data-dependent algorithms can be categorized into two main groups: unsupervised and supervised. Unlike unsupervised hashing \cite{YAR}  \cite{BT} \cite{YSAF} that explores data intrinsic structures to preserve similarity relations between neighbors without any supervision, supervised hashing \cite{MDM} \cite{JSF} \cite{CAMP} employs semantic information to learn hash functions, and thus it usually achieves better retrieval accuracy than unsupervised hashing on semantic similarity measures.     

Among supervised hashing algorithms, pairwise based hashing, maintaining the relationship of similar or dissimilar pairs in a Hamming space, is one popular method to exploit label information. Numerous pairwise based algorithms have been proposed in the past decade, including spectral hashing (SH) \cite{YAR}, linear discriminant analysis hashing (LDAHash) \cite{CAMP}, minimal loss hashing (MLH) \cite{MDM},  binary reconstruction embedding (BRE) \cite{BT} and kernel-based supervised hashing (KSH) \cite{WJR}, etc. Although these algorithms have been demonstrated effective in many large-scale tasks, their employed optimization strategies are usually insufficient to explore the similarity information defined in the non-convex and non-differential objective functions. In order to handle these non-smooth and non-convex problems, four main strategies have been proposed: symmetric/asymmetric relaxation, and asymmetric/symmetric discrete. Symmetric relaxation \cite{YAR} \cite{JSF} \cite{WJR} \cite{GCD} is to relax discrete binary vectors in a continuous feasible region followed by thresholding to obtain binary codes. Although symmetric relaxation can simplify the original optimization problem, it often generates large accumulated errors between hash and linear functions. To reduce the accumulated error, asymmetric relaxation \cite{shi2016kernel} utilizes the element-wise product of discrete and its relaxed continuous matrices to approximate a pairwise label matrix. Asymmetric discrete hashing \cite{neyshabur2013power} \cite{FWSY} \cite{shi2017asymmetric} usually utilizes the product of two distinct discrete matrices to preserve pair relations into a binary space. Symmetric discrete hashing \cite{xia2014supervised} \cite{kang2016column} firstly learns binary codes with preserving symmetric discrete constraints and then trains classifiers based on the learned discrete codes. Although most of hashing algorithms with these four strategies have achieved promising performance, they have at least one of the following four major disadvantages: (i) Learning binary codes employs relaxation and thresholding strategies, thereby producing large accumulated errors; (ii) Learning binary codes requires high storage and computation costs, i.e. \(\mathcal{O}(n^2)\), where \(n\) is the number of data points, thereby limiting their applications to large-scale tasks; (iii) The used pairwise label matrix usually emphasizes the difference of images among different classes but neglects their relevance within the same class. Hence, existing optimization methods might perform poorly to preserve the relevance information among images; (iv) The employed optimization methods focus on one type of optimization problems and it is difficult to directly apply them to other problems. 

Motivated by aforementioned observations, in this paper, we propose a novel simple, general and scalable optimization method that can solve various pairwise based hashing models for directly learning binary codes. The main contributions are summarized as follows:
\begin{itemize}
\item We propose a novel alternative optimization mechanism to reformulate one typical quartic problem, in term of hash functions in the original objective of KSH \cite{WJR}, into a linear problem by introducing a linear regression model. 

\item We present and analyze a significant discovery that gradually updating each batch of binary codes in a sequential mode, i.e. batch by batch, is greatly beneficial to the convergence of binary code learning. 

\item We propose a scalable symmetric discrete hashing algorithm with gradually updating each batch of one discrete matrix. To make the update step more smooth, we further present a greedy symmetric discrete hashing algorithm to greedily update each bit of batch discrete matrices. Then we demonstrate that the proposed greedy hashing algorithm can be used to solve other optimization problems in pairwise based hashing.

\item Extensive experiments on three benchmark databases: CIFAR-10 \cite{ARW}, NUS-WIDE \cite{chua2009nus} and COCO \cite{lin2014microsoft}, demonstrate the superior performance of the proposed method over recent state of the arts, with low time costs. 
\end{itemize}

\section{Related Work}
\label{relatedwork}
Based on the manner of using the label information, supervised hashing can be classified into three major categories: point-wise, multi-wise and pairwise.

\textbf{Point-wise based hashing} formulates the searching into one classification problem based on the rule that the classification accuracy with learned binary codes should be maximized. Supervised discrete hashing (SDH) \cite{FCWH} leverages one linear regression model to generate optimal binary codes. Fast supervised discrete hashing (FSDH) \cite{koutaki2016fast} improves the computation requirement of SDH via fast SDH approximation. 
Supervised quantization hashing (SQH) \cite{wang2016supervised} introduces composite quantization into a linear model to further boost the discriminative ability of binary codes. Deep learning of binary hash codes (DLBHC) \cite{lin2015deep} and deep supervised convolutional hashing (DSCH) \cite{sapkota2018deep} employ convolutional neural network to simultaneously learn image representations and hash codes in a point-wised manner. Point-wise based hashing is scalable and its optimization problem is relatively easier than multi-wise and pairwise based hashing; however, its rule is inferior compared to the other two types of supervised hashing. 

\textbf{Multi-wise based hashing} is also named as ranking based hashing that learns hash functions to maximize the agreement of similarity orders over two items between original and Hamming distances. Triplet ranking hashing (TRH) \cite{norouzi2012hamming} and column generation hashing (CGH) \cite{li2013learning} utilize a triplet ranking loss to maximumly preserve the similarity order.  Order preserving hashing (OPH) \cite{wang2013order} learns hash functions to maximumly preserve the similarity order by taking it as a classification problem. Ranking based supervised hashing (RSH) \cite{wang2013learning} constructs a ranking triplet matrix to maintain orders of ranking lists. Ranking preserving hashing (RPH) \cite{wang2015ranking} learns hash functions by directly optimizing a ranking measure, Normalized Discounted Cumulative Gain (NDCG) \cite{jarvelin2000ir}. Top rank supervised binary coding (Top-RSBC) \cite{song2015top} focuses on boosting the precision of top positions in a Hamming distance ranking list. Discrete semantic ranking hashing (DSeRH) \cite{liu2017discretely} learns hash functions to maintain ranking orders with preserving symmetric discrete constraints. Deep network in network hashing (DNNH) \cite{lai2015simultaneous}, deep semantic ranking hashing (DSRH) \cite{zhao2015deep} and triplet-based deep binary embedding (TDBE) \cite{zhuang2016fast} utilize convolutional neural network to learn image representations and hash codes based on the triplet ranking loss over three items. Most of these multi-wise based hashing algorithms relax the ranking order or discrete binary codes in a continuously feasible region to solve their original non-convex and non-smooth problems. 

\textbf{Pairwise based hashing} maintains relationship among originally high-dimensional data into a Hamming space by calculating and preserving the relationship of each pair. SH \cite{YAR} constructs one graph to maintain the similarity among neighbors and then utilizes it to map the high-dimensional data into a low-dimensional Hamming space. Although the original version of SH is unsupervised hashing, it is easily converted into a supervised algorithm. Inspired by SH, many variants including anchor graph hashing \cite{WJS}, elastic embedding \cite{carreira2010elastic}, discrete graph hashing (DGH) \cite{WMSS}, and asymmetric discrete graph hashing (ADGH) \cite{shi2017asymmetric} have been proposed. LDAHash \cite{CAMP} projects the high-dimensional descriptor vectors into a low-dimensional Hamming space with maximizing the distance of  inter-class data and meanwhile minimizing the intra-class distances. MLH \cite{MDM} adopts a structured prediction with latent variables to learn hash functions. BRE \cite{BT} aims to minimize the difference between Euclidean distances of original data and their Hamming distances. It leverages a coordinate-descent algorithm to solve the optimization problem with preserving symmetric discrete constraints. SSH \cite{JSF} introduces a pairwise matrix and KSH \cite{WJR} leverages the Hamming distance between pairs to approximate the pairwise matrix. This objective function is intuitive and simple, but the optimization problem is highly non-differential and difficult to directly solve. KSH utilizes a ``symmetric relaxation + greedy'' strategy to solve the problem. Two-step hashing (TSH) \cite{GCD} and FastHash \cite{lin2015supervised} relax the discrete constraints into a continuous region \(\left [ -1,1 \right ]\). Kernel based discrete supervised hashing (KSDH) \cite{shi2016kernel} adopts asymmetric relaxation to simultaneously learn the discrete matrix and a low-dimensional projection matrix for hash functions. Lin: Lin and Lin: V \cite{neyshabur2013power}, asymmetric inner-product binary coding (AIBC) \cite{FWSY} and asymmetric discrete graph hashing (ADGH) \cite{shi2017asymmetric} employ the asymmetric discrete mechanism to learn low-dimensional matrices. Column sampling based discrete supervised hashing (COSDISH) \cite{kang2016column} adopts the column sampling strategy same as latent factor hashing (LFH) \cite{zhang2014supervised} but directly learn binary codes by reformulating the binary quadratic programming (BQP) problems into equivalent clustering problems. Convolutional neural network hashing (CNNH) \cite{xia2014supervised} divide the optimization problem into two sub-problems \cite{zhang2010self}: (i) learning binary codes by a coordinate descent algorithm using Newton directions; (ii) training a convolutional neural network using the learned binary codes as labels. After that, deep hashing network (DHN) \cite{zhu2016deep} and deep supervised pairwise hashing (DSPH) \cite{li2015feature} simultaneously learn image representations and binary codes using pairwise labels. HashNet \cite{cao2017hashnet} learns binary codes from imbalanced similarity data. Deep cauchy hashing (DCH)  \cite{cao2018deep}  utilizes pairwise labels to generate compact and concentrated binary codes for efficient and
effective Hamming space retrieval. Unlike previous work, in this paper we aim to present a simpler, more general and scalable optimization method for binary code learning.

\section{Symmetric Discrete Hashing via A Pairwise Matrix}
\label{DOPH}
In this paper, matrices and vectors are represented by boldface uppercase and lowercase letters, respectively. For a matrix \(\mathbf{X} \in \mathbb{R}^{n\times d}\), its \(i\)-th row and \(j\)-th column vectors are denoted as \(\mathbf{x}_{i}\) and \(\mathbf{x}^{j}\), respectively, and \(x_{ij}\) is one entry at the \(i\)-th row and \(j\)-th column.
\vspace{-0.5em}
\subsection{Formulation}
KSH \cite{WJR} is one popular pairwise based hashing algorithm, which can preserve pairs' relationship with using two identical binary matrices to approximate one pairwise real matrix. Additionally, it is a quartic optimization problem in term of hash functions, and thus more typical and difficult to solve than that only containing a quadratic term with respect to hash functions. Therefore, we first propose a novel optimization mechanism to solve the original problem in KSH, and then extend the proposed method to solve other pairwise based hashing models.

Given \(n\) data points \(\mathbf{X}=\left [ \mathbf{x}_{1},\mathbf{x}_{2},\cdots,\mathbf{x}_{n} \right ]\in \mathbb{R}^{n\times d}\), suppose one pair \((\mathbf{x}_{i},\mathbf{x}_{j})\in \mathcal{M}\) when they are neighbors in a metric space or share at least one common label, and \((\mathbf{x}_{i},\mathbf{x}_{j})\in \mathcal{C}\) when they are non-neighbors in a metric space or have different class labels. For the single-label multi-class problem, the pairwise matrix \(\mathbf{S}\in \mathbb{R}^{n\times n}\) is defined as \cite{JSF}:
\begin{equation}
\vspace{-0.2em}
s_{ij}=\left\{\begin{matrix}
1 & (x_{i},x_{j})\in \mathcal{M},\\ 
-1 &  (x_{i},x_{j})\in \mathcal{C},\\ 
0 & otherwise. 
\end{matrix}\right.
\label{eqn:sc}
\vspace{-0.2em}
\end{equation}
For the multi-label multi-class problem, similar to \cite{shi2018pairwise}, \(\mathbf{S}\) can be defined as:
\begin{equation}
\vspace{-0.2em}
s_{ij}=\left\{\begin{matrix}
r_{ij} & (x_{i},x_{j})\in \mathcal{M},\\ 
\alpha &  (x_{i},x_{j})\in \mathcal{C},\\ 
0 & otherwise,
\end{matrix}\right.
\label{eqn:sm}
\vspace{-0.2em}
\end{equation}
where \(r_{ij}>0\) is the relevance between \(\mathbf{x}_{i}\) and \(\mathbf{x}_{j}\), which is defined as the number of common labels shared by \(\mathbf{x}_{i}\) and \(\mathbf{x}_{j}\). \(\alpha<0\) is the weight to describe the difference between \(\mathbf{x}_{i}\) and \(\mathbf{x}_{j}\). In this paper, to preserve the difference between non-neighbor pairs, we empirically set \(\alpha=-\frac{r_{max}}{2}\), where \(r_{max}\) is the maximum relevance among all neighbor pairs. We do not set \(\alpha=-r_{max}\) because few data pairs have the relevance being \(r_{max}\).

To encode one data point \(\mathbf{x} \in \mathbb{R}^{d}\) into \(m\)-bit hash codes, its \(k\)-th hash function can be defined as:
\begin{equation}
\vspace{-0.2em}
h_{k}(\mathbf{x})=sgn(\mathbf{x}\mathbf{a}_{k}^T+b_{k}),
\label{eqn:hk}
\vspace{-0.2em}
\end{equation}
where \(\mathbf{a}_{k}\in \mathbb{R}^{d}\) is a projection vector, and \(sgn(\mathbf{x}_{i}\mathbf{a}_{k}^T+b_{k})=1\) if \(\mathbf{x}_{i}\mathbf{a}_{k}^T+b_{k}\geq 0\) , otherwise \(sgn(\mathbf{x}_{i}\mathbf{a}_{k}^T+b_{k})=-1\). Note that since \(\mathbf{x}\mathbf{a}_{k}^T+b_{k}\) can be written as the form \(\mathbf{x}\mathbf{a}_{k}^T\) with \(\mathbf{x}\) adding one dimension and \(\mathbf{a}_{k}\) absorbing \(b_{k}\), for simplicity we utilize \(h_{k}(\mathbf{x})=sgn(\mathbf{x}\mathbf{a}_{k}^T)\) in this paper. Let \(code_{m}(\mathbf{x})=\left \{ h_{1}(\mathbf{x}),h_{2}(\mathbf{x}),\cdots,h_{m}(\mathbf{x})\right \}\) be hash codes of \(\mathbf{x}\), and then for any pair \((\mathbf{x}_{i}, \mathbf{x}_{j})\), we have \(m\geq code_{m}(\mathbf{x}_{i})\circ code_{m}(\mathbf{x}_{j})\geq -m\). To approximate the pairwise matrix \(\mathbf{S}\), same as \cite{WJR}, a least-squares style objective function is defined as:
\begin{equation}
\vspace{-0.2em}
\underset{\mathbf{A}}{min}\left \| \mathbf{H}\mathbf{H}^T-\lambda \mathbf{S} \right \|_F^2,\ s.t.\ \mathbf{H}=sgn(\mathbf{X}\mathbf{A}^T),
\label{eqn:obj}
\vspace{-0.2em}
\end{equation}
where \(\lambda=\frac{m}{r_{max}}\) and \(\mathbf{A}\in \mathbb{R}^{m\times d}\) is a low-dimensional projection matrix. Eq. (\ref{eqn:obj}) is a quartic problem in term of hash functions, and this can be demonstrated by expanding its objective function.
\vspace{-0.5em}
\subsection{Symmetric Discrete Hashing}
\subsubsection{Formulation transformation}
In this subsection, we show the procedure to transform Eq. (\ref{eqn:obj}) into a linear problem. Since the objective function in Eq. (\ref{eqn:obj}) is a highly non-differential quartic problem in term of hash functions \(sgn(\mathbf{X}\mathbf{A}^T)\), it is difficult to directly solve this problem. Here, we solve the problem in Eq. (\ref{eqn:obj}) via a novel alternative optimization mechanism: reformulating the quartic problem in term of hash functions into a quadratic one and then linearizing the quadratic problem. We present the detailed procedure in the following.

Firstly, we introduce a Lemma to show one of our main motivations to transform the quartic problem into a linear problem.  
\newtheorem{namedtheorem}{Lemma}
\begin{namedtheorem}
When the matrix \(\mathbf{A} \in \mathbb{R}^{m\times d}\) satisfies the condition: \(\mathbf{X}\mathbf{A}^T=\mathbf{Y}\), it is a global solution of the following problem:
\begin{equation}
\underset{\mathbf{A}}{max}\ Tr\left \{ \mathbf{H}^T\mathbf{Y} \right \},\ s.t.\ \mathbf{H}=sgn(\mathbf{X}\mathbf{A}^T).
\label{eqn:lemma}
\end{equation}
\end{namedtheorem}
Lemma 1 is easy to solve, because when \(\mathbf{X}\mathbf{A}^T=\mathbf{Y}\), \(\mathbf{H}=sgn(\mathbf{X}\mathbf{A}^T)=sgn(\mathbf{Y})\) makes the objective in Eq. (\ref{eqn:lemma}) attain the maximum.  \emph{Since \(\mathbf{A}\) satisfying \(\mathbf{X}\mathbf{A}^T=\mathbf{Y}\) is a global solution of the problem in Eq. (\ref{eqn:lemma}), it suggests that the problem in term of hash functions can be transformed into a linear problem in term of \(\mathbf{A}\).} Inspired by this observation, we can solve the quartic problem in term of hash functions. For brevity, in the following we first ignore the constraint \(\mathbf{H}=sgn(\mathbf{X}\mathbf{A}^T)\) in Eq. (\ref{eqn:obj}) and aim to transform the quartic problem in term of \(\mathbf{H}\) into the linear form as the objective in Eq. (\ref{eqn:lemma}), and then obtain the low-dimensional projection matrix \(\mathbf{A}\).

To reformulate the quartic problem in term of \(\mathbf{H}\) into a quadratic one, in the \(l\)-th iteration, we set one discrete matrix to be \(\mathbf{H}_{l-1}\) and aim to solve the following quadratic problem in term of \(\mathbf{H}\):
\begin{equation}
\underset{\mathbf{H}}{min}\left \| \mathbf{H}_{l-1}\mathbf{H}^{T}-\lambda \mathbf{S} \right \|_F^2,\ s.t.\  \mathbf{H}\in \left \{ -1,1 \right \}^{n\times m}.
\label{eqn:HHl}
\end{equation}
Note that the problem in Eq. (\ref{eqn:HHl}) is not strictly equal to the problem in Eq. (\ref{eqn:obj}) w.r.t \(\mathbf{H}\). However, when \(\mathbf{H}_{l}=\mathbf{H}_{l-1}\), it is the optimal solution of both Eq. (\ref{eqn:HHl}) and Eq. (\ref{eqn:obj}) w.r.t \(\mathbf{H}\). The details are shown in Proposition 1.

\newtheorem{ptheorem}{Proposition}
\begin{ptheorem}
When \(\mathbf{H}_{l}=\mathbf{H}_{l-1}\), the optimal solution of Eq. (\ref{eqn:HHl}) is also the optimal solution of Eq. (\ref{eqn:obj}) w.r.t \(\mathbf{H}\).
\end{ptheorem}

\begin{proof} 
Obviously, if \(\mathbf{H}_{l}=\mathbf{H}_{l-1}\), it is the optimal solution of Eq. (\ref{eqn:HHl}). Then we can consider the following formulation:
\begin{equation}
\underset{\mathbf{H}_{l},\mathbf{H}_{l-1}}{min} \left \| \mathbf{H}_{l}\mathbf{H}_{l-1}-\mathbf{S} \right \|_F^2 \leq \underset{\mathbf{H}}{min} \left \| \mathbf{H}\mathbf{H}-\mathbf{S}\right \|_F^2 
\label{eqn:p1}
\end{equation}
Similar to one major motivation of asymmetric discrete hashing algorithms \cite{neyshabur2013power} \cite{shi2017asymmetric}, in Eq. (\ref{eqn:p1}), the feasible region of \(\mathbf{H}_{l}\), \(\mathbf{H}_{l-1}\) in the left term is more flexible than \(\mathbf{H}\) in the right term (Eq. (\ref{eqn:obj})), i.e. the left term contains both two cases \(\mathbf{H}_{l}\neq \mathbf{H}_{l-1}\) and  \(\mathbf{H}_{l}= \mathbf{H}_{l-1}\). Only when \(\mathbf{H}_{l}= \mathbf{H}_{l-1}\), \(\underset{\mathbf{H}_{l}, \mathbf{H}_{l-1}}{min}\left \| \mathbf{H}_{l}\mathbf{H}_{l-1}-\mathbf{S} \right \|_F^2 = \underset{\mathbf{H}}{min} \left \| \mathbf{H}\mathbf{H}-\mathbf{S}\right \|_F^2\). It suggests that when \(\mathbf{H}_{l}=\mathbf{H}_{l-1}\), it is the optimal solution of Eq. (\ref{eqn:obj}). Therefore, when \(\mathbf{H}_{l}=\mathbf{H}_{l-1}\), it is the optimal solution of both Eq. (\ref{eqn:obj}) and Eq. (\ref{eqn:HHl}).
\end{proof}
Inspired by Proposition 1, we aim to seek \(\mathbf{H}_{l}=\mathbf{H}_{l-1}\) through solving the problem in Eq. (\ref{eqn:HHl}). Because \(\lambda\) is known and \(Tr\left \{\mathbf{S}^T\mathbf{S}\right \}=constant\), the optimization problem in Eq. (\ref{eqn:HHl}) equals:
\begin{equation}
\begin{array}{cc}
\underset{\mathbf{H}}{min}\ Tr\left \{ \mathbf{H}\mathbf{H}_{l-1}^{T}\mathbf{H}_{l-1}\mathbf{H}^T\right \} -2\lambda Tr\left \{  \mathbf{H}\mathbf{H}_{l-1}^{T}\mathbf{S}\right \},\\
s.t.\  \mathbf{H}\in \left \{ -1,1 \right \}^{n\times m}.
\end{array}
\label{eqn:eHHl}
\end{equation}
Since \(Tr\left \{\mathbf{H}\mathbf{H}_{l-1}^{T}\mathbf{S}\right \}\) is a linear problem in term of \(\mathbf{H}\), the main difficulty to solve Eq. (\ref{eqn:eHHl}) is caused by the non-convex quadratic term \(Tr\left \{ \mathbf{H}\mathbf{H}_{l-1}^{T}\mathbf{H}_{l-1}\mathbf{H}^T\right \}\). Thus we aim to linearize this quadratic term in term of \(\mathbf{H}\) by introducing a linear regression model as follows:

\newtheorem{theorem}{Theorem}
\begin{theorem}
Given a discrete matrix \(\mathbf{H}\in 
\left \{ -1,1 \right \}^{n\times m}\) and one real nonzero matrix \(\mathbf{Z}\in \mathbb{R}^{m\times m}\), \(inf\left \{ \left \| \mathbf{H}-\mathbf{P}\mathbf{Z} \right \|_F^2+\left \|\mathbf{P}\mathbf{\Gamma}^{\frac{1}{2}}\right \|_F^2|\mathbf{P}\in \mathbb{R}^{n\times m}, \Gamma_{ii}> 0 \right \}=Tr\left \{ \mathbf{H}(\mathbf{I}_{m}-\mathbf{Z}^T(\mathbf{Z}\mathbf{Z}^T+\mathbf{\Gamma})^{-1}\mathbf{Z})\mathbf{H}^T \right \}\), where \(\mathbf{\Gamma}\in \mathbb{R}^{m\times m}\) is a diagonal matrix and \(\mathbf{I}_{m}\in \mathbb{R}^{m\times m}\) is an identity matrix.
\end{theorem}
\begin{proof}
It is easy to verify that \(\mathbf{P}^{\ast}=\mathbf{H}\mathbf{Z}^T(\mathbf{Z}\mathbf{Z}^T+\mathbf{\Gamma})^{-1}\) is the global optimal solution to the problem \(\underset{\mathbf{P}}{min}\left \| \mathbf{H}-\mathbf{P}\mathbf{Z} \right \|_F^2+\left \|\mathbf{P}\mathbf{\Gamma}^{\frac{1}{2}}\right \|_F^2\). Substituting \(\mathbf{P}^{\ast}\) into the above objective, its minimum value is \(Tr\left \{ \mathbf{H} (\mathbf{I}_{m}-\mathbf{Z}^T(\mathbf{Z}\mathbf{Z}^T+\mathbf{\Gamma})^{-1}\mathbf{Z})\mathbf{H}^T \right \}\).
Therefore, Theorem 1 is proved.
\end{proof} 
Theorem 1 suggests that when \(\mathbf{H}_{l-1}^{T}\mathbf{H}_{l-1}=\gamma (\mathbf{I}_{m}-\mathbf{Z}^T(\mathbf{Z}\mathbf{Z}^T+\mathbf{\Gamma})^{-1}\mathbf{Z})\), the quadratic problem in Eq. (\ref{eqn:eHHl}) can be linearized as a regression type. We show the details in Theorem 2.
\begin{theorem}
When \(\mathbf{H}_{l-1}^{T}\mathbf{H}_{l-1}=\gamma (\mathbf{I}_{m}-\mathbf{Z}^T(\mathbf{Z}\mathbf{Z}^T+\mathbf{\Gamma})^{-1}\mathbf{Z})\), where \(\gamma\) is a constant, the problem in Eq. (\ref{eqn:eHHl}) can be reformulated as:
\begin{equation}
\begin{array}{cc}
\underset{\mathbf{H}, \mathbf{P}}{min}\ \gamma(\left \| \mathbf{H}-\mathbf{P}\mathbf{Z} \right \|_F^2+\left \| \mathbf{P} \mathbf{\Gamma}^{\frac{1}{2}} \right \|_F^2)\\-2\lambda Tr\left \{\mathbf{H}\mathbf{H}_{l-1}^{T}\mathbf{S}^T  \right \}, s.t.\ \mathbf{H}\in \left \{ -1,1 \right \}^{n\times m}.  
\end{array}
\label{eqn:rsh}
\end{equation}
\end{theorem}
\begin{proof}
Based on the condition \(\mathbf{H}_{l-1}^{T}\mathbf{H}_{l-1}=\gamma (\mathbf{I}_{m}-\mathbf{Z}^T(\mathbf{Z}\mathbf{Z}^T+\mathbf{\Gamma})^{-1}\mathbf{Z})\) and Theorem 1, substituting \(\mathbf{P}^{\ast}=\mathbf{H}\mathbf{Z}^T(\mathbf{Z}\mathbf{Z}^T+ \mathbf{\Gamma})^{-1}\) into the objective of Eq. (\ref{eqn:rsh}), whose objective value is equal to that of Eq. (\ref{eqn:eHHl}), Therefore, Theorem 2 is proved.
\end{proof}
Since \(Tr\left \{\mathbf{H}^T\mathbf{H}\right \}=mn\), the problem in Eq. (\ref{eqn:rsh}) equals:
\begin{equation}
\begin{array}{ccc}
\underset{\mathbf{H},\mathbf{P}}{max}\, Tr\left \{\mathbf{H}(\frac{\lambda}{\gamma}\mathbf{H}_{l-1}^{T}\mathbf{S}+\mathbf{Z}^T\mathbf{P}^T)\right \}\\-\frac{1}{2}(\left \| \mathbf{P}\mathbf{\Gamma}^{\frac{1}{2}} \right\|_{F}^2 +Tr\left \{\mathbf{P}\mathbf{Z}\mathbf{Z}^T\mathbf{P}^T\right \}),\\ 
\ s.t.\ \mathbf{H}\in \left \{ -1,1 \right \}^{n\times m},
\end{array}
\label{eqn:lrsh}
\end{equation}
which is a linear problem in term of \(\mathbf{H}\).
\begin{figure*}[ht]
\center
\subfigure[@ \(\left \| \mathbf{H}_{l}-\mathbf{H}_{l-1} \right \|_{F}\)]{\includegraphics[trim={0em 0em 0em 0em}, clip, width=0.28\textwidth]{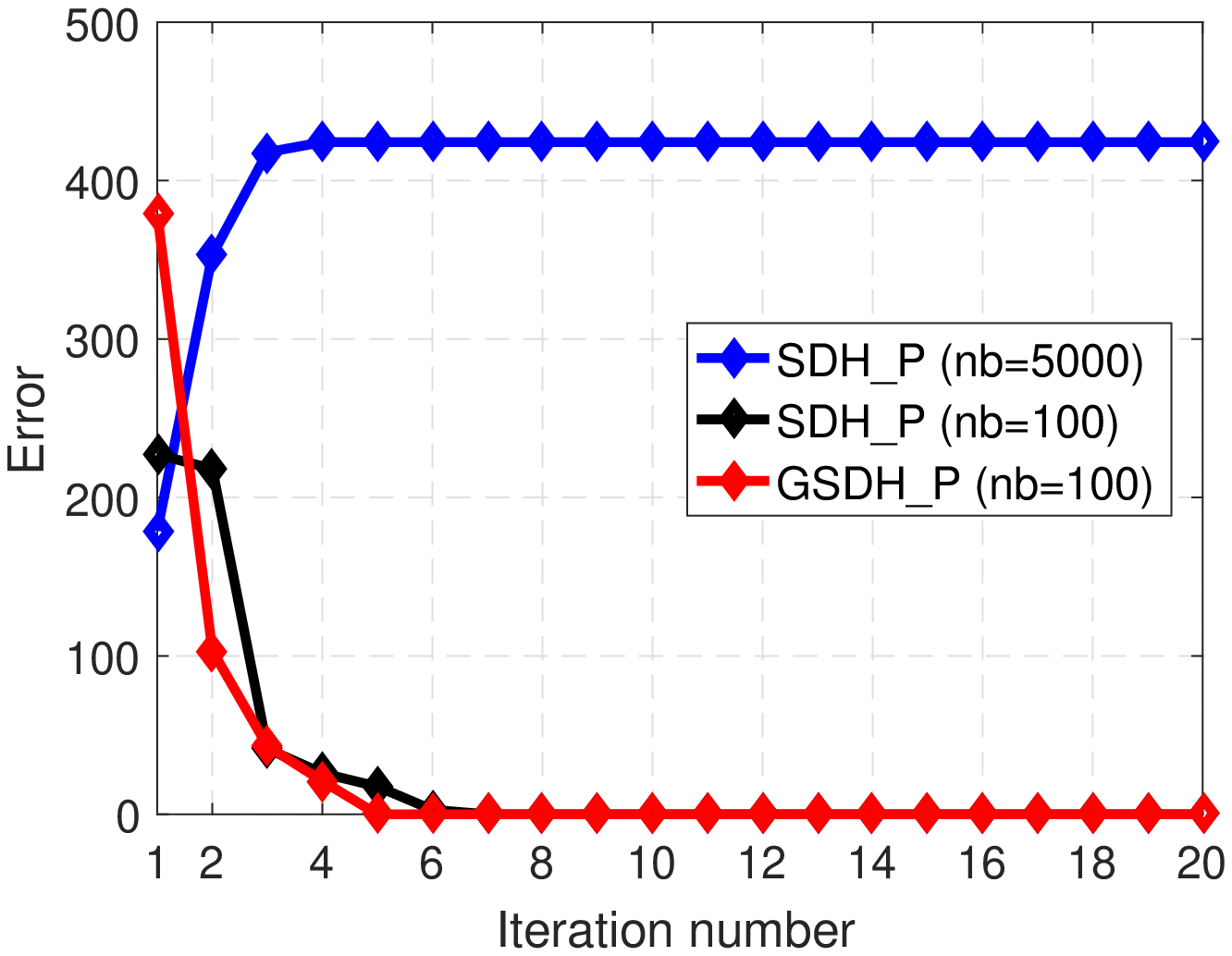}}
\subfigure[@ \(\left \| \mathbf{H}\mathbf{H}^T-\lambda\mathbf{S}\right \|_{F}\)]{\includegraphics[trim={0em 0em 0em 0em}, clip, width=0.28\textwidth]{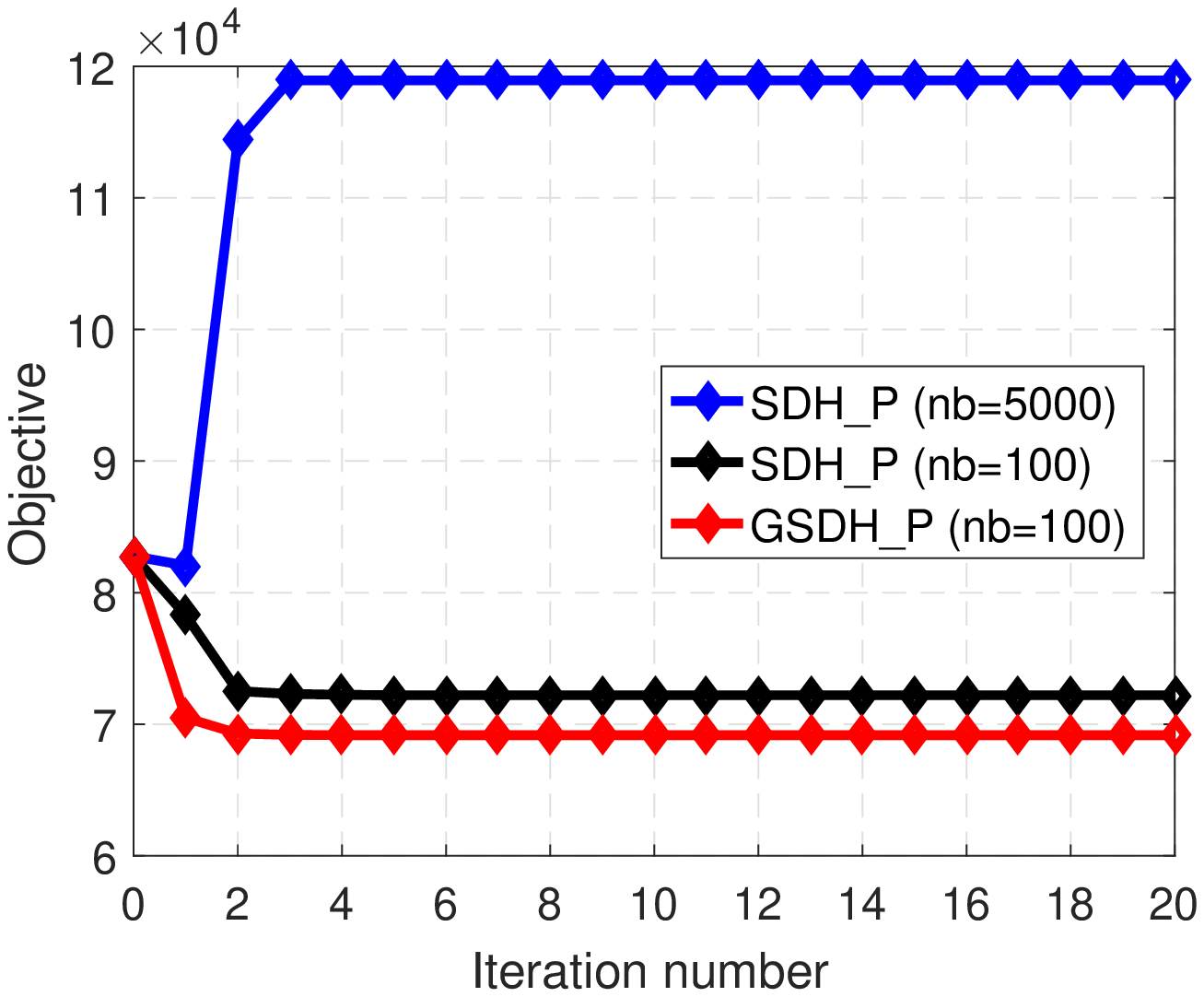}}
\subfigure[@ Top 500 ]{\includegraphics[trim={0em 0em 0em 0em}, clip, width=0.28\textwidth]{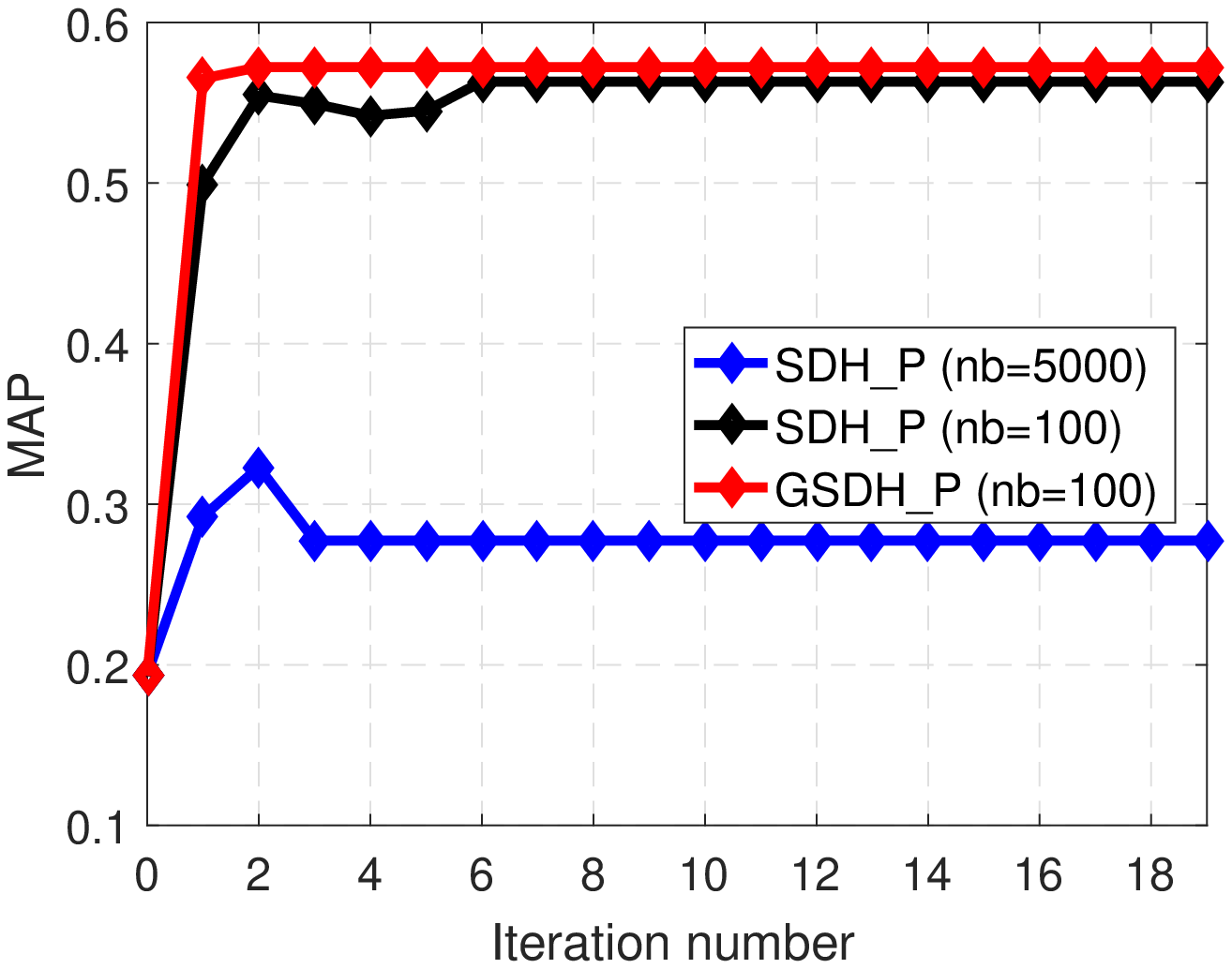}}
\vspace{-0.7em}
\caption{The error, objective and mean average precision (MAP) of SDH\_P and GSDH\_P with different number of iterations. (In total we select 5K training and 1K query images from the CIFAR10 database, and employ all training samples as anchors. In (b) and (c), when the number of iteration is 0, the results are achieved by using the projection matrix calculated in the initialization step.)}
\label{figobj}
\vspace{-1em}
\end{figure*}

Next, we demonstrate that there exists \(\gamma\) and \(\mathbf{Z}\) such that \(\mathbf{H}_{l-1}^{T}\mathbf{H}_{l-1}=\gamma (\mathbf{I}_{m}-\mathbf{Z}^T(\mathbf{Z}\mathbf{Z}^T+\mathbf{\Gamma})^{-1}\mathbf{Z})\). The details are shown in Theorem 3.
\begin{theorem}
Suppose that a full rank matrix \(\mathbf{H}_{l-1}^{T}\mathbf{H}_{l-1}=\mathbf{U}\mathbf{\Lambda}^2\mathbf{U}^T\), where \(\mathbf{\Lambda}\in \mathbb{R}^{m\times m}\) is a positive diagonal matrix and \(\mathbf{U}^T\mathbf{U}=\mathbf{U}\mathbf{U}^T=\mathbf{I}_{m}\). If \(\gamma \geq \Lambda_{ii}^2\) and \(\Gamma_{ii}>0\), \((1\leq i\leq m)\) and a real nonzero matrix \(\mathbf{Z}=\mathbf{V} \mathbf{\Delta} \mathbf{U}^T\) satisfies the conditions: \(\mathbf{V}^T\mathbf{V}=\mathbf{V}\mathbf{V}^T=\mathbf{I}_{m}\), and \(\mathbf{\Delta}\in \mathbb{R}^{m\times m}\) is a non-negative real diagonal matrix with the \(i\)-th diagonal element being \(\Delta_{ii}=\sqrt{\frac{\gamma\Gamma_{ii}} {\Lambda_{ii}^2}-\Gamma_{ii}}\).
\end{theorem}
\begin{proof}
Based on singular value decomposition (SVD), there exist matrices \(\mathbf{V}\) and \(\mathbf{U}_{\mathbf{Z}}\), satisfying the conditions \(\mathbf{V}\mathbf{V}^T=\mathbf{V}^T\mathbf{V}=\mathbf{I}_{m}\) and \(\mathbf{U}_{\mathbf{Z}}\mathbf{U}_{\mathbf{Z}}^T=\mathbf{U}_{\mathbf{Z}}^T\mathbf{U}_{\mathbf{Z}}=\mathbf{I}_{m}\), such that a real nonzero matrix \(\mathbf{Z}\) is represented by \(\mathbf{Z}=\mathbf{V}\mathbf{\Delta}\mathbf{U}_\mathbf{Z}^T\), where \(\mathbf{\Delta}\) is a non-negative real diagonal matrix.  Then \(\mathbf{I}_{m}-\mathbf{Z}^T(\mathbf{Z}\mathbf{Z}^T+\mathbf{\Gamma})^{-1}\mathbf{Z}=\mathbf{U}_{\mathbf{Z}}(\mathbf{I}_{m}-\mathbf{\Delta}(\mathbf{\Delta}^2+\mathbf{\Gamma})^{-1}\mathbf{\Delta})\mathbf{U}_{\mathbf{Z}}^T\). Note that when the vectors in \(\mathbf{V}\) and \(\mathbf{U}_{\mathbf{Z}}\) corresponds to the zero diagonal elements, they can be constructed by employing a Gram-Schmidt process such that \(\mathbf{V}\mathbf{V}^T=\mathbf{V}^T\mathbf{V}=\mathbf{I}_{m}\) and \(\mathbf{U}_{\mathbf{Z}}\mathbf{U}_{\mathbf{Z}}^T=\mathbf{U}_{\mathbf{Z}}^T\mathbf{U}_{\mathbf{Z}}=\mathbf{I}_{m}\), and these constructed vectors are not unique.

Since \(\mathbf{H}_{l-1}^{T}\mathbf{H}_{l-1}=\mathbf{U}\mathbf{\Lambda}^2\mathbf{U}^T\) and \(\mathbf{I}_{m}-\mathbf{Z}^T(\mathbf{Z}\mathbf{Z}^T+\mathbf{\Gamma})^{-1}\mathbf{Z}=\frac{\mathbf{H}_{l-1}^T\mathbf{H}_{l-1}}{\gamma}\), it can have \(\gamma\Gamma_{ii}(\Delta_{ii}^2+\Gamma_{ii})^{-1}=\Lambda_{ii}^2\) when \(\mathbf{U}_{\mathbf{Z}}=\mathbf{U}\). Since there exists \(0 < \Gamma_{ii}(\Delta_{ii}^2+\Gamma_{ii})^{-1}\leq 1\), \(\gamma\) should satisfy: \(\gamma\geq \Lambda_{ii}^2\) and \(\Gamma_{ii}>0\). Additionally, based on \(\gamma\Gamma_{ii}(\Delta_{ii}^2+\Gamma_{ii})^{-1}=\Lambda_{ii}^2\), there exists \(\Delta_{ii}=\sqrt{\frac{\gamma \Gamma_{ii}} {\Lambda_{ii}^2}-\Gamma_{ii}}\).
Therefore, Theorem 3 is proved.
\end{proof}

\(\mathbf{H}_{l-1}^T\mathbf{H}_{l-1}\) is usually a positive-definite matrix thanks to \(m<<n\), leading to  \(\Lambda_{ii}>0\). Based on Theorem 3, it is easy to construct a real nonzero matrix \(\mathbf{Z}\). Since \(\gamma\geq \Lambda_{ii}^2\), we set \(\gamma=\underset{i}{max}\ \Lambda_{ii}^2+\beta\) for simplicity, where \(1\leq i\leq m\) and \(\beta\geq 0\) is a constant.  Then Eq. (\ref{eqn:lrsh}) can be solved by alternatively updating \(\mathbf{H}\), \(\mathbf{Z}\) and \(\mathbf{P}\). Actually, we can obtain \(\mathbf{H}\) by using an efficient algorithm in Theorem 4 that does not need to compute the matrices \(\mathbf{Z}\) and \(\mathbf{P}\).
\begin{theorem}
For the inner \(t\)-th iteration embedded in the outer \(l\)-th iteration, the problem in Eq. (\ref{eqn:lrsh}) can be reformulated as the following problem:
\begin{equation}
\vspace{-0.2em}
\begin{array}{cc}
\underset{\mathbf{H}_{l}}{max}\,Tr\left\{\mathbf{H}_{l}((\gamma \mathbf{I}_{m}-\mathbf{H}_{l-1}^{T}\mathbf{H}_{l-1})\mathbf{H}_{l_{t-1}}^{T} \right. \\ \left.
+\lambda\mathbf{H}_{l-1}^{T}\mathbf{S})\right \},
s.t. \ \mathbf{H}_{l}\in \left \{ -1,1 \right \}^{n\times m},  
\end{array} 
\label{eqn:al}
\vspace{-0.2em}
\end{equation}
where \(\mathbf{H}_{l}\in \mathbb{R}^{n\times m}\) denotes binary codes \(\mathbf{H}\) in the outer \(l\)-th iteration, and \(\mathbf{H}_{l_{t-1}}\) represents the obtained binary codes \(\mathbf{H}\) at the inner \(t-1\)-th iteration embedded in the outer \(l\)-th iteration.
\end{theorem}
\begin{proof}
In Eq. (\ref{eqn:lrsh}), for the inner \(t\)-th iteration embedded in the outer \(l\)-th iteration, fixing \(\mathbf{H}\) as \(\mathbf{H}_{l_{t-1}}\), it is easy to obtain \(\mathbf{P}_{l_{t}}=\mathbf{H}_{l_{t-1}}\mathbf{Z}^T(\mathbf{Z}\mathbf{Z}^T+\mathbf{\Gamma})^{-1}\). Substituting \(\mathbf{P}_{l_{t}}\) into Eq. (\ref{eqn:lrsh}), it becomes:
\begin{equation}
\begin{array}{ccc}
\underset{\mathbf{H}_{l}}{max}\ Tr\left \{\mathbf{H}_{l}\mathbf{Z}^T (\mathbf{Z}\mathbf{Z}^T+\mathbf{\Gamma})^{-1}\mathbf{Z} \mathbf{H}_{l_{t-1}}^{T}\right \} \\
+\frac{\lambda}{\gamma} Tr\left \{\mathbf{H}_{l}\mathbf{H}_{l-1}^{T}\mathbf{S}\right \},\ s.t. \ \mathbf{H}_{l}\in \left \{ -1,1 \right \}^{n\times m}. 
\end{array}
\label{eqn:ha}
\end{equation}
Based on Theorem 2 and its proof, there have \(\gamma \mathbf{I}_{m}-\mathbf{H}_{l-1}^{T}\mathbf{H}_{l-1}=\gamma \mathbf{Z}^T (\mathbf{Z}\mathbf{Z}^T+\mathbf{\Gamma})^{-1}\mathbf{Z}\). Substituting it into Eq. (\ref{eqn:ha}), whose optimization problem becomes Eq. (\ref{eqn:al}). Therefore, Theorem 4 is proved. 
\end{proof} 

For the inner loop embedded in the outer \(l\)-th iteration, there are many choices for the initialization value \(\mathbf{H}_{l_{0}}\). Here, we set \(\mathbf{H}_{l_{0}}=\mathbf{H}_{l-1}\). At the \(t\)-th iteration, the global solution of Eq. (\ref{eqn:al}) is \(\mathbf{H}_{l_{t}}=sgn(\lambda\mathbf{S}^T\mathbf{H}_{l-1}+\mathbf{H}_{l_{t-1}}(\gamma \mathbf{I}_{m}-\mathbf{H}_{l-1}^T\mathbf{H}_{l-1}))\). Additionally, for the inner loop, both \(\mathbf{P}_{l_{t}}\) and \(\mathbf{H}_{l_{t}}\) are global solutions of the \(t\)-th iteration, it suggests that the objective of Eq. (\ref{eqn:al}) will be non-decreasing and converge to at least a local optima. Therefore, we have the following theorem.
\begin{theorem}
For the inner loop embedded in the outer \(l\)-th iteration, the objective of Eq. (\ref{eqn:al}) is monotonically non-decreasing in each iteration and will converge to at least a local optima.
\end{theorem}

Although Theorem 5 suggests that the objective of Eq. (\ref{eqn:al}) will converge, its convergence is largely affected by the parameter \(\gamma\), which is used to balance the convergence and semantic information in \(\mathbf{S}\). Usually, the larger \(\gamma\), the faster convergence but the more loss of semantic information. Therefore, we empirically set a small non-negative constant for \(\beta\), i.e. \(0 \leq \beta \leq 100\), where \(\gamma=\underset{i}{max}\ \Lambda_{ii}^2+\beta\).

%
Based on Eq. (\ref{eqn:al}), the optimal solution \(\mathbf{H}_{l}^{\ast}\) can be obtained. Then we utilize \(\mathbf{H}_{l}^{\ast}\) to replace \(\mathbf{H}_{l-1}\) in Eq. (\ref{eqn:HHl}) for next iteration in order to obtain the optimal solution \(\mathbf{H}^{\ast}\). Since \(\mathbf{H}^{\ast}=sgn(\lambda\mathbf{S}^T\mathbf{H}^{\ast}+\mathbf{H}^{\ast}(\gamma \mathbf{I}_{m}-\mathbf{H}^{\ast T}\mathbf{H}^{\ast}))\) and \(\mathbf{H}=sgn(\mathbf{X}\mathbf{A}^T)\), based on Lemma 1, \(\mathbf{A}\) should satisfy  \(\mathbf{X}\mathbf{A}^T=\lambda\mathbf{S}^T\mathbf{H}^{\ast}+\mathbf{H}^{\ast}(\gamma \mathbf{I}_{m}-\mathbf{H}^{\ast T}\mathbf{H}^{\ast})\). However, it is an overdetermined linear system due to \(n>>d\). For simplicity, we utilize a least-squares model to obtain \(\mathbf{A}\), which is \(\mathbf{A}= (\lambda\mathbf{H}^{\ast T}\mathbf{S}+(\gamma \mathbf{I}_{m}-\mathbf{H}^{\ast T}\mathbf{H}^{\ast})\mathbf{H}^{\ast T}
)\mathbf{X}(\mathbf{X}^T\mathbf{X})^{-1}\).
\vspace{-0.5em}
\subsubsection{Scalable symmetric discrete hashing with updating batch binary codes}

\emph{\textbf{Remark}: The optimal solution of Eq. (\ref{eqn:HHl}) is at least the local optimal solution of Eq. (\ref{eqn:obj}) only when \(\left \| \mathbf{H}_{l}-\mathbf{H}_{l-1} \right \|_{F}=0\). Given an initialization \(\mathbf{H}^{0}\), \(\mathbf{H}\) can be alternatively updated by solving Eq. (\ref{eqn:al}). However, with updating all binary codes at once on the non-convex feasible region, \(\mathbf{H}\) might change on two different discrete matrices, which would lead to the error \(\left \| \mathbf{H}_{l}-\mathbf{H}_{l-1} \right \|_{F} \neq 0\) (please see Figure \ref{figobj}a) and the objective of Eq. (\ref{eqn:obj}) becomes worse (please see Figure \ref{figobj}b). Therefore, we divide \(\mathbf{H}\) into a variety of batches and gradually update each of them in a sequential mode, i.e. batch by batch.}

To update one batch of \(\mathbf{H}\), i.e. \(\mathbf{H}_{b}=\mathbf{H}(idx,:)\), where \(idx \in \mathbb{R}^{n_{b}}\) is one column vector denoting the index of selected binary codes in \(\mathbf{H}\), the optimization problem derived from Eq. (\ref{eqn:HHl}) is:
\begin{equation}
\begin{array}{cc}
\underset{\mathbf{H}_{b}}{min}\left \| \mathbf{H}_{l-1}\mathbf{H}_{b}^{T}-\lambda \mathbf{S}_{b} \right \|_F^2,\\ s.t.\  \mathbf{H}_{b}\in \left \{ -1,1 \right \}^{n_{b}\times m}, \mathbf{H}_{b}\subset \mathbf{H},
\end{array}
\label{eqn:bHHl}
\end{equation}
where \(\mathbf{S}_{b}=\mathbf{S}(:,idx) \in \mathbb{R}^{n\times n_{b}}\). 

Furthermore, although \(\mathbf{S}\in \mathbb{R}^{n\times n}\) is high-dimensional for large \(n\), it is low-rank or can be approximated as a low-rank matrix. Similar to previous algorithms \cite{WJS} \cite{shi2017asymmetric}, we can select \(p\) (\(p<<n\)) samples from \(n\) training samples as anchors and then construct an anchor based pairwise matrix \(\mathbf{S}_{A}\in \mathbb{R}^{p\times n}\), which preserves almost all similarity information of \(\mathbf{S}\). Let \(\mathbf{H}_{A} \in \left \{ -1,1 \right \}^{p\times m} \) denote binary codes of anchors, and then utilize \(\mathbf{S}_{A}\) to replace \(\mathbf{S}\) for updating \(\mathbf{H}_{b}\), Eq. (\ref{eqn:bHHl}) becomes:
 \begin{equation}
 \begin{array}{cc}
 \underset{\mathbf{H}_{b}}{min}\left \| \mathbf{H}_{Al-1}\mathbf{H}_{b}^{T}-\lambda \mathbf{S}_{Ab} \right \|_F^2,\\ s.t.\  \mathbf{H}_{b}\in \left \{ -1,1 \right \}^{n_{b}\times m}, \mathbf{H}_{A}\subset \mathbf{H}, \mathbf{H}_{b}\subset \mathbf{H},
 \end{array}
 \label{eqn:AbHHl}
 \end{equation}
 where \(\mathbf{H}_{Al-1}\) denotes \(\mathbf{H}_{A}\) obtained at the \(l-1\)-th iteration in the outer loop, and \(\mathbf{S}_{Ab}=\mathbf{S}_{A}(:,idx) \in \mathbb{R}^{p\times n_{b}}\). 
\begin{table}[tbp]
\begin{tabular}{l}
\hline
\textbf{Algorithm 1: SDH\_P}         \\
\hline
\textbf{Input:} Data \(\mathbf{X}\in \mathbb{R}^{n\times d}\), pairwise matrix \(\mathbf{S}_A \in \mathbb{R}^{p\times n}\),\\  bit number \(m\), parameters \(\lambda\), \(\beta>0\), batch size \(n_{b}\),\\ anchor index \(a\_idx \in \mathbb{R}^{p}\),
outer and inner loop \\ maximum  iteration number \(L_{1}\), \(L_{2}\). \\
\textbf{Output:} \(\mathbf{A}\in \mathbb{R}^{m\times d}\) and \(\mathbf{H}\in \left \{ -1,1 \right \}^{n\times m}\). \\
\hline
1:\ \textbf{Initialize:} Let \(\mathbf{X}_{A}=\mathbf{X}(a\_idx,:)\), set \(\mathbf{A}\) to be the \\
 \quad left-eigenvectors of  \(\mathbf{X}^T\mathbf{S}_{A}^T\mathbf{X}_{A}\) corresponding to \\ \quad   its largest \(m\) eigenvalues, calculate \(\mathbf{H}=sgn(\mathbf{X}\mathbf{A}^T)\) \\ \quad  and \(\mathbf{H}_A=\mathbf{H}(a\_idx,:)\).\\
2:\ \textbf{while not converge or reach maximum iterations}  \\
3.\ \quad \(index\leftarrow randperm(n)\);  \\
4.\ \quad \textbf{for} \(i=1\) to \(\frac{n}{n_{b}}\) \textbf{do}  \\
5.\ \qquad \(idx\leftarrow index((i-1)n_{b}+1:in_{b})\);\\
6:\ \qquad Do the SVD of \(\mathbf{H}_A^{T}\mathbf{H}_A=\mathbf{U}\mathbf{\Lambda}^{2}\mathbf{U}^T\); \\
7:\ \qquad \(\gamma \leftarrow max(diag(\mathbf{\Lambda}^2))+\beta\);\\
8:\ \qquad \textbf{repeat}  \\
9:\ \qquad \qquad \(\mathbf{H}(idx,:)\leftarrow sgn(\lambda\mathbf{S}_A(:,idx)^T\mathbf{H}_A\)\\ \quad \qquad \qquad \(+\mathbf{H}(idx,:)(\gamma\mathbf{I}_{m}-\mathbf{H}_A^T\mathbf{H}_A))\); \\
10:\qquad \textbf{until convergence} \\
11: \ \qquad \(\mathbf{H}_A=\mathbf{H}(a\_idx,:)\); \\
12:\ \quad \textbf{end for} \\
13:\ \textbf{end while} \\
14: Do the SVD of \(\mathbf{H}_A^{T}\mathbf{H}_A=\mathbf{U}\mathbf{\Lambda}^{2}\mathbf{U}^T\); \\
15: \(\gamma \leftarrow max(diag(\mathbf{\Lambda}))+\beta\);\\
16: \(\mathbf{A}= (\lambda\mathbf{H}_A^{T}\mathbf{S}_A+(\gamma \mathbf{I}_{m}-\mathbf{H}_A^{T}\mathbf{H}_A)\mathbf{H}^{T}
)\mathbf{X}(\mathbf{X}^T\mathbf{X})^{-1}\).\\
\hline
\end{tabular}
\vspace{-1em}
\end{table}

\begin{table}[tb]
\begin{tabular}{l}
\hline
\textbf{Algorithm 2: GSDH\_P}         \\
\hline
\textbf{Input:} Data \(\mathbf{X}\in \mathbb{R}^{n\times d}\), pairwise matrix \(\mathbf{S}_{A} \in \mathbb{R}^{p\times n}\),\\  bit number \(m\), parameters \(\lambda\), \(\beta>0\), batch size \(n_{b}\),\\ anchor index \(a\_idx \in \mathbb{R}^{p}\),
outer/inner maximum \\ iteration number \(L_{1}\), \(L_{2}\). \\
\textbf{Output:} \(\mathbf{A}\in \mathbb{R}^{m\times d}\) and \(\mathbf{H}\in \left \{ -1,1 \right \}^{n\times m}\). \\
\hline
1:\ \textbf{Initialize:} Let \(\mathbf{X}_{A}=\mathbf{X}(a\_idx,:)\), set \(\mathbf{A}\) to be the \\
 \quad left-eigenvectors of  \(\mathbf{X}^T\mathbf{S}_{A}^T\mathbf{X}_{A}\) corresponding to \\ \quad   its largest \(m\) eigenvalues, calculate \(\mathbf{H}=sgn(\mathbf{X}\mathbf{A}^T)\), \\ \quad and \(\mathbf{H}_A=\mathbf{H}(a\_idx,:)\).\\
2:\ \textbf{while not converge or reach maximum iterations}  \\
3.\ \quad \(index\leftarrow randperm(n)\);  \\
4.\ \quad \textbf{for} \(i=1\) to \(\frac{n}{n_{b}}\) \textbf{do}  \\
5.\ \qquad \(idx\leftarrow index((i-1)n_{b}+1:in_{b})\);\\
6:\ \qquad  \textbf{for} \(j=1\) to \(m\) \textbf{do}  \\
7:\ \qquad \quad Calculating \(\mathbf{\widetilde{S}}_A(:,idx)\) with fixing \(\mathbf{h}^{k}\),\\ \qquad \qquad \(k=1,2,\cdots,m\), \(k\neq j\);  \\
8:\ \qquad\quad \textbf{repeat}  \\
9:\ \qquad \quad \qquad  \(\mathbf{H}(idx,j)\leftarrow sgn(\lambda\mathbf{\widetilde{S}}_A(:,idx)^T\mathbf{H}_A(:,j)\)\\ \qquad \qquad \qquad \(+\beta  \mathbf{H}(idx,j))\); \\
10:\qquad \quad\textbf{until convergence} \\
11: \qquad \quad \(\mathbf{H}_A(:,j)=\mathbf{H}(a\_idx,j)\); \\
12:\qquad \textbf{end for} \\
13: \quad \textbf{end for} \\
14: \textbf{end while} \\
15: Do the SVD of \(\mathbf{H}_A^{T}\mathbf{H}_A=\mathbf{U}\mathbf{\Lambda}^{2}\mathbf{U}^T\); \\
16: \(\gamma \leftarrow max(diag(\mathbf{\Lambda}))+\beta\);\\
17: \(\mathbf{A}= (\lambda\mathbf{H}_A^{T}\mathbf{S}_A+(\gamma \mathbf{I}_{m}-\mathbf{H}_A^{T}\mathbf{H}_A)\mathbf{H}^{T}
)\mathbf{X}(\mathbf{X}^T\mathbf{X})^{-1}\).\\

\hline
\end{tabular}
\vspace{-1em}
\end{table}

Similar to Eq. (\ref{eqn:HHl}), the problem in Eq. (\ref{eqn:AbHHl}) can be firstly transformed into a quadratic problem, and then can be reformulated as a similar form to Eq. (\ref{eqn:al}) based on Theorems 1-4. e.g.
\begin{equation}
\vspace{-0.2em}
\begin{array}{cc}
\underset{\mathbf{H}_{bl}}{max}\,Tr\left\{\mathbf{H}_{bl}((\gamma \mathbf{I}_{m}-\mathbf{H}_{Al-1}^{T}\mathbf{H}_{Al-1})\mathbf{H}_{bl_{t-1}}^{T} \right. \\ \left.
+\lambda\mathbf{H}_{Al-1}^{T}\mathbf{S}_{Ab})\right \},
s.t. \ \mathbf{H}_{bl}\in \left \{ -1,1 \right \}^{n_{b}\times m}.  
\end{array} 
\label{eqn:al1}
\vspace{-0.2em}
\end{equation}
where \(\mathbf{H}_{bl}\) denotes the batch binary codes at the \(l\)-th iteration in the outer loop. 

For clarity, we present the detailed optimization procedure to attain \(\mathbf{H}\) by updating each batch \(\mathbf{H}_{b}\) and calculate the projection matrix \(\mathbf{A}\) in Algorithm 1, namely symmetric discrete hashing via a pairwise matrix (SDH\_P). For Algorithm 1, with gradually updating each batch of \(\mathbf{H}\), the error \(\left \| \mathbf{H}_{l}-\mathbf{H}_{l-1} \right \|_{F}\) usually converges to zero (please see Figure \ref{figobj}a) and the objective of Eq. (\ref{eqn:obj}) also converges to a better local optima (please see Figure \ref{figobj}b). Besides, we also display the retrieval performance in term of mean average precision (MAP) with a small batch size and different iterations in Figure \ref{figobj}c.

\vspace{-1em}
\subsection{Greedy Symmetric Discrete Hashing}
To make the update step more smooth, we greedily update each bit of the batch matrix \(\mathbf{H}_{b}\). Suppose that \(\mathbf{h}_{b}^{j}=\mathbf{H}(idx,j)\) is the \(j\)-th bit of \(\mathbf{H}_{b}\), it can be updated by solving the following optimization problem:
\begin{equation}
\vspace{-0.2em}
\begin{array}{cc}
\underset{\mathbf{h}_{b}^{j}}{min}\left \| \mathbf{h}_A^{j}\mathbf{h}_{b}^{jT}-\lambda (\mathbf{S}_{Ab}-\sum_{k\neq j}^{m} \mathbf{h}_A^{k}\mathbf{h}_{b}^{kT})  \right \|_F^2,\\ s.t.\  \mathbf{h}_{b}^{j}\in \left \{ -1,1 \right \}^{n_{b}},  \mathbf{h}_{A}^{j}\subset \mathbf{h}^{j}, \mathbf{h}_{b}^j\subset \mathbf{h}^j,
\end{array}
\label{eqn:bjHHl}
\vspace{-0.2em}
\end{equation}
where \(\mathbf{h}_A^{j}\) and \(\mathbf{h}_A^{k}\) represent the \(j\)-th and \(k\)-th bits of \(\mathbf{H}_A\), respectively, and \(\mathbf{h}_{b}^{k}\) is the \(k\)-th bit of \(\mathbf{H}_{b}\).

The problem in Eq. (\ref{eqn:bjHHl}) can also be firstly transformed into a quadratic problem and then solved using Theorems 1-4. Similar to Eq. (\ref{eqn:al}), the problem in Eq. (\ref{eqn:bjHHl}) can be transformed to:
 
\begin{equation}
\vspace{-0.2em}
\begin{array}{cc}
\underset{\mathbf{h}_{bl}^{j}}{max}\,Tr\left\{\mathbf{h}_{bl}^{j}(\beta \mathbf{h}_{bl_{t-1}}^{jT}+\lambda\mathbf{h}_{Al-1}^{jT}\mathbf{\widetilde{S}}_{Ab})\right \}, \\
s.t. \ \mathbf{h}_{bl}^{j}\in \left \{ -1,1 \right \}^{n_{b}},  
\end{array} 
\label{eqn:gal1}
\vspace{-0.2em}
\end{equation}
\noindent where \(\mathbf{h}_{bl}^j\) is the \(\mathbf{h}_{b}^j\) obtained at the \(l\)-th iteration in the outer loop, \(\mathbf{h}_{bl_{t-1}}^{j}\) is the \(\mathbf{h}_{b}^j\) obtained at the \(t-1\)-th iteration in inner loop embedded in the \(l\)-th outer loop  and \(\mathbf{\widetilde{S}}_{Ab}= \mathbf{S}_{Ab}-\sum_{k\neq j}^{m} \mathbf{h}_A^{k}\mathbf{h}_{b}^{kT}\).

In summary, we show the detailed optimization procedure in Algorithm 2, namely greedy symmetric discrete hashing via a pairwise matrix (GSDH\_P). The error \(\left \| \mathbf{H}_{l}-\mathbf{H}_{l-1} \right \|_{F}\) and the objective of Eq. (\ref{eqn:obj}) in Algorithm 2 and its retrieval performance in term of MAP with different number of iterations are shown in Figure \ref{figobj}a, b and c, respectively.

\textbf{Out-of-sample:} In the query stage, \(\mathbf{H}\) is employed as the binary codes of training data. We adopt two strategies to encode the query data point \(\mathbf{q}\in \mathbb{R}^{d}\): (i) encoding it using \(h(\mathbf{q})=sgn(\mathbf{q}\mathbf{A}^T)\); (ii) similar to previous algorithms \cite{lin2014fast} \cite{xia2014supervised}, employing \(\mathbf{H}\) as labels to learn a classification model, like least-squares, decision trees (DT) or convolutional neural networks (CNNs), to classify \(\mathbf{q}\). 

\subsection{Convergence Analysis}
Empirically, when \(n>>n_{b}\), the proposed algorithms can converge to at least a local optima, although they cannot be theoretically guaranteed to converge in all cases. Here, we explain why gradually updating each batch of binary codes is beneficial to the convergence of hash code learning. 

In Eq. (\ref{eqn:al}), with updating one batch of \(\mathbf{H}\), i.e. \(\mathbf{H}_{b}\in \left \{ -1,1 \right \}^{n_b\times m} \), Eq. (\ref{eqn:al}) becomes:
\begin{equation}
\vspace{-0.2em}
\begin{array}{cc}
\underset{\mathbf{H}_{bl}}{max}\,Tr\left\{\mathbf{H}_{bl}((\gamma \mathbf{I}_{m}-\mathbf{H}_{l-1}^{T}\mathbf{H}_{l-1})\mathbf{H}_{bl_{t-1}}^{T} \right. \\ \left.
+\lambda\mathbf{H}_{l-1}^{T}\mathbf{S}_b)\right \},
s.t. \ \mathbf{H}_{bl} \in \left \{ -1,1 \right \}^{n_{b}\times m},  
\end{array} 
\label{eqn:alhb}
\vspace{-0.2em}
\end{equation} 
The hash code matrix \(\mathbf{H}\) can be represented as \(\mathbf{H}=\left [ \mathbf{H}_b; \mathbf{\widetilde{H}} \right ]\), where  \(\mathbf{\widetilde{H}}\in  \left \{ -1,1 \right \}^{(n-n_b)\times m}\). Since \(n>>n_b\), the objective of Eq. (\ref{eqn:alhb}) is determined by:
\begin{equation}
\vspace{-0.2em}
\begin{array}{cc}
\underset{\mathbf{H}_{bl}}{max}\,Tr\left\{\mathbf{H}_{bl}((\gamma \mathbf{I}_{m}-\mathbf{\widetilde{H}}_{l-1}^{T}\mathbf{\widetilde{H}}_{l-1})\mathbf{H}_{bl_{t-1}}^{T} \right. \\ \left.
+\lambda\mathbf{\widetilde{H}}_{l-1}^{T}\mathbf{S}_b)\right \},
s.t. \ \mathbf{H}_{bl} \in \left \{ -1,1 \right \}^{n_{b}\times m},  
\end{array} 
\label{eqn:malhb}
\vspace{-0.2em}
\end{equation} 

Based on Theorem 5, the inner loop can theoretically guarantee the convergence of the objective in Eq. (\ref{eqn:al}), and thus the optimal solution \(\mathbf{H}_{bl}^{\ast}\) of Eq. (\ref{eqn:malhb}) can be obtained by the inner loop. Then it has:
\begin{equation}
\begin{array}{ll}
Tr\left\{\mathbf{H}_{bl}^{\ast}((\gamma \mathbf{I}_{m}-\mathbf{\widetilde{H}}_{l-1}^{T}\mathbf{\widetilde{H}}_{l-1})\mathbf{H}_{bl}^{\ast T}+\lambda\mathbf{\widetilde{H}}_{l-1}^{T}\mathbf{S}_b)\right \} \geq \\ Tr\left\{\mathbf{H}_{bl-1}((\gamma \mathbf{I}_{m}-\mathbf{\widetilde{H}}_{l-1}^{T}\mathbf{\widetilde{H}}_{l-1})\mathbf{H}_{bl-1}^{T}+\lambda\mathbf{\widetilde{H}}_{l-1}^{T}\mathbf{S}_b)\right \}
\end{array}
\label{eqn:Hl1}
\end{equation}
Because of \(n>>n_{b}\), Eq. (\ref{eqn:Hl1}) usually leads to  
\begin{equation}
\begin{array}{ll}
Tr\left\{\mathbf{H}_{bl}^{\ast}((\gamma \mathbf{I}_{m}- \mathbf{\widehat{H}}_{l}^T\mathbf{\widehat{H}}_{l})\mathbf{H}_{bl}^{\ast T} +\lambda \mathbf{\widehat{H}}_{l}^{T}\mathbf{S}_b)\right \} \geq \\ Tr\left\{\mathbf{H}_{bl-1}((\gamma \mathbf{I}_{m}-\mathbf{H}_{l-1}^{T}\mathbf{H}_{l-1})\mathbf{H}_{bl-1}^{T}+\lambda\mathbf{H}_{l-1}^{T}\mathbf{S}_b)\right \}
\end{array}
\label{eqn:hbl}
\end{equation}
where \(\mathbf{\widehat{H}}_{l}=\left [ \mathbf{H}_{bl}^{\ast}; \mathbf{\widetilde{H}}_{l-1} \right ]\) and \(\mathbf{H}_{l-1}=\left [ \mathbf{H}_{bl-1}; \mathbf{\widetilde{H}}_{l-1} \right ]\). Eq. (\ref{eqn:hbl}) suggests that when \(n_{b}<<n\), updating each batch matrix can usually make the objective of Eq. (\ref{eqn:al}) gradually converge to at least a local optima. 

\subsection{Time Complexity Analysis}
In Algorithm 1, \(n>>d>>m\) and \(n>>p>>m\). Step 1 calculating matrices \(\mathbf{A}\) and \(\mathbf{H}\) requires \(\mathcal{O}(pnd)\) and \(\mathcal{O}(ndm)\) operations, respectively. For the outer loop stage, the time complexity of steps 6, 7, 9 and 11 is \(\mathcal{O}(pm^2)\), \(\mathcal{O}(m)\), \(\mathcal{O}(pmn_{b})\) and \(\mathcal{O}(pm)\), respectively. Hence, the outer loop stage spends \(\mathcal{O}(L_{1}L_{2}npm)\) operations. Steps 14-16 to calculate the projection matrix \(\mathbf{A}\) spend \(\mathcal{O}(pm^2)\), \(\mathcal{O}(m)\) and \(max(\mathcal{O}(pnm), \mathcal{O}(nd^2))\), respectively. Therefore, the total complexity of Algorithm 1 is \(max(\mathcal{O}(npd), \mathcal{O}(nd^2), \mathcal{O}(L_{1}L_{2}npm))\). Empirically, \(L_{1}\leq 20\) and \(L_{1}\leq 3\). 

For Algorithm 2, step 1 calculating  \(\mathbf{A}\), \(\mathbf{H}\) and \(\mathbf{M}\) spends at most \(max(\mathcal{O}(pnd), \mathcal{O}(nd^2))\). In the loop stage, the major steps both 7 and 9 require \(\mathcal{O}(pn_{b})\) operations. Hence, the total time complexity of the loop stage is \(\mathcal{O}(L_{1}L_{2}npm)\). Additionally, calculating the final \(\mathbf{A}\) costs the same time to the steps 14-16 in Algorithm 1. Therefore, the time complexity of Algorithm 2 is \(max(\mathcal{O}(npd), \mathcal{O}(nd^2), \mathcal{O}(L_{1}L_{2}npm))\). 

\section{Extension to Other Hashing Algorithms}
In this subsection, we illustrate that the proposed algorithm GSDH\_P is suitable for solving many other pairwise based hashing models.

Two-step hashing algorithms \cite{GCD} \cite{lin2015supervised} iteratively update each bit of the different loss functions defined on the Hamming distance of data pairs so that the loss functions of many hashing algorithms such as BRE \cite{BT}, MLH \cite{MDM} and EE \cite{carreira2010elastic} are incorporated into a general framework, which can be written as:
\begin{equation}
\underset{\mathbf{h}^j}{min}\ \mathbf{h}^{j}\mathbf{L}\mathbf{h}^{jT},\ s.t.\ \mathbf{h}^j\in \left \{ -1,1 \right \}^n 
\label{eqn:Lh}
\end{equation}
where \(\mathbf{h}^j\) represents the \(j\)-th bit of binary codes \(\mathbf{H}\in\left \{ -1,1 \right \}^{n\times m}\) and \(\mathbf{L}\in \mathbb{R}^{n\times n}\) is obtained based on different loss functions with fixing all bits of binary codes except \(\mathbf{h}^j\).

The algorithms \cite{GCD} \cite{lin2015supervised} firstly relax \(\mathbf{h}^j\in \left \{ -1,1 \right \}^n\) into \(\mathbf{h}^j\in \left [ -1,1 \right ]^n\) and then employ L-BFGS-B \cite{zhu1997algorithm} to solve the relaxed optimization problem, followed by thresholding to attain the binary vector \(\mathbf{h}^j\). However, our optimization mechanism can directly solve Eq. (\ref{eqn:Lh}) without relaxing \(\mathbf{h}^j\). 

Since \(Tr(\mathbf{h}^j\mathbf{h}^{jT}\mathbf{h}^j\mathbf{h}^{jT})=const\) and \(Tr(\mathbf{L}\mathbf{L}^T)=const\), the problem in Eq. (\ref{eqn:Lh}) can be equivalently reformulated as:
\begin{equation}
\underset{\mathbf{h}^j}{min} \left \| \mathbf{h}^j\mathbf{h}^{jT}-(-\mathbf{L}) \right \|_F^2\ s.t.\ \mathbf{h}^j\in \left \{ -1,1 \right \}^n, 
\label{eqn:hL}
\end{equation}
whose optimization type is same as the objective of Eq. (\ref{eqn:obj}) w.r.t \(\mathbf{H}\). Replacing the constraint \(\mathbf{h}^j\in \left \{ -1,1 \right \}^n\) with \(\mathbf{h}^j=sgn(\mathbf{X}\mathbf{a}_{j}^T)\), where \(\mathbf{a}_{j}\in \mathbb{R}^{d}\) is the \(j\)-th row vector of \(\mathbf{A}\), Eq. (\ref{eqn:hL}) becomes:
\begin{equation}
\underset{\mathbf{a}_{j}}{min} \left \| \mathbf{h}^j\mathbf{h}^{jT}-(-\mathbf{L})  \right \|_F^2\ s.t.\ \mathbf{h}^j=sgn(\mathbf{X}\mathbf{a}_{j}^T),
\label{eqn:fL}
\end{equation}
Since \(\mathbf{L}\in \mathbb{R}^{n\times n}\) will consume large computation and storage costs for large \(n\), we select \(p\) training anchors to construct \(\mathbf{L}_{A}\in \mathbb{R}^{p\times n}\) based on different loss functions. Replacing \(\mathbf{L}\) in Eq. (\ref{eqn:fL}) with \(\mathbf{L}_{A}\), it becomes:

\begin{equation}
\underset{\mathbf{a}_{j}}{min} \left \| \mathbf{h}_{A}^j\mathbf{h}^{jT}-(-\mathbf{L}_{A})  \right \|_F^2,\  s.t.\ \mathbf{h}^j=sgn(\mathbf{X}\mathbf{a}_{j}^T),
\label{eqn:fLA}
\end{equation}

Similar to solving Eq. (\ref{eqn:obj}), we can firstly obtain \(\mathbf{h}^{j}\) and then calculate \(\mathbf{a}_{j}\). To attain \(\mathbf{h}^{j}\), we still gradually update each batch \(\mathbf{h}_{b}^{j}\) by solving the following problem:
\begin{equation}
\begin{array}{cc}
\underset{\mathbf{h}_b^{j}}{min} \left \| \mathbf{h}_{A}^j\mathbf{h}_{b}^{jT}-(-\mathbf{L}_{Ab})  \right \|_F^2,
\end{array}
\label{eqn:fLAb}
\end{equation}
where \(\mathbf{L}_{Ab}=\mathbf{L}_{A}(:,idx)\in \mathbb{R}^{p\times n_{b}}\).

The optimization type of Eq. (\ref{eqn:fLAb}) is the same as that of Eq. (\ref{eqn:bjHHl}). \(\mathbf{h}^{j}\) can be obtained by gradually updating \(\mathbf{h}_{b}^{j}\) as shown in GSDH\_P.  After obtaining \(\mathbf{h}^{j}\),  \(\mathbf{a}_{j}\) can be attained by using \(\mathbf{a}_{j}= (\beta \mathbf{h}^{jT} -\mathbf{h}_{A}^{jT}\mathbf{L}_{A} )\mathbf{X}(\mathbf{X}^T\mathbf{X})^{-1}\).

   

Based on Eqs. (\ref{eqn:Lh})-(\ref{eqn:fLAb}), many pairwise based hashing models can lean binary codes by using GSDH\_P. For instances, we show the performance on solving the optimization model in BRE \cite{BT} \cite{lin2015supervised}:
\begin{equation}
\mathcal{L}(\mathbf{h}_{i}, \mathbf{h}_{j})= \left [ m\delta (s_{ij}<0)-d_{H}(\mathbf{h}_{i},\mathbf{h}_{j}) \right ]^2
\label{eqn:bre}
\end{equation}
where \(d_{H}(\mathbf{h}_{i} ,\mathbf{h}_{j})\) is the Hamming distance between \(\mathbf{h}_{i}\) and \(\mathbf{h}_{j}\), and \(\delta (\cdot)\in \left \{ 0,1 \right \}\) is an indicator function. Here, \(s_{ij}<0\) denotes \((\mathbf{x}_{i},\mathbf{x}_{j})\in \mathcal{C}\). \\
\noindent One typical model with a hinge loss function \cite{MDM} \cite{lin2015supervised}:
\begin{equation}
\mathcal{L}(\mathbf{h}_{i}, \mathbf{h}_{j})= \left\{\begin{matrix}
\left [ 0-d_{H}(\mathbf{h}_{i} ,\mathbf{h}_{j}) \right ]^2  & \ (\mathbf{x}_{i},\mathbf{x}_{j})\in \mathcal{M} \\ 
 \left [ max(0.5m-d_{H}(\mathbf{h}_{i} ,\mathbf{h}_{j}), 0) \right ]^2  & otherwise.
\end{matrix}\right. 
\label{eqn:hinge}
\end{equation}
 In this paper, the optimization model in BRE solved by GSDH\_P is named as GSDH\_P\(_{BRE}\). Similarly, the hinge loss function Eq. (\ref{eqn:hinge}) solved by GSDH\_P is named as GSDH\_P\(_{Hinge}\).

\section{Experimental Results and Analysis}
\label{experiments}
We evaluate the proposed algorithms SDH\_P and GSDH\_P on one benchmark single-label database: CIFAR-10 \cite{ARW}, and two popular multi-label databases: NUS-WIDE \cite{chua2009nus} and COCO \cite{lin2014microsoft}. The CIFAR-10 database contains 60K color images of ten object categories, with each category consisting of 6K images. The NUS-WIDE database contains around 270K images collected from Flickr, with each image consisting of multiple semantic labels. Totally, this database has 81 ground truth concept labels. Here, similar to \cite{FCWH}, we choose the images associated with the 21 most frequent labels. In total, there are around 195K images. The COCO database consists of about 328K images belonging to 91 objects types. We adopt the 2014 training and validation datasets. They have around 83K training and 41K validation images belonging to 80 object categories.

\subsection{Experimental Setting}
 We compare SDH\_P and GSDH\_P against fifteen state-of-the-art hashing methods including six point-wise and pairwise based algorithms: KSH \cite{WJR}, CCA-ITQ \cite{gong2013iterative}, SDH \cite{FCWH}, COSDISH \cite{kang2016column}, KSDH \cite{shi2016kernel}, ADGH \cite{shi2017asymmetric}, four ranking algorithms: RSH \cite{wang2013learning}, CGH \cite{li2013learning}, Top-RSBC \cite{song2015top} and DSeRH \cite{liu2017discretely}, and five deep hashing algorithms: CNNH \cite{xia2014supervised}, DNNH \cite{lai2015simultaneous}, DHN \cite{zhu2016deep}, Hashnet \cite{cao2017hashnet} and DCH \cite{cao2018deep}. Additionally, we also compare BRE \cite{BT} and MLH  \cite{MDM} and FastH \cite{lin2015supervised} with GSDH\_P\(_{BRE}\) and GSDH\_P\(_{Hinge}\) to illustrate the generalization of GSDH\_P. For KSH and KSDH, we employ the same pairwise matrix as SDH\_P and GSDH\_P for tackling single-label and multi-label tasks. For all deep hashing algorithms, we show their reported retrieval accuracy on each database. Additionally, for better comparing with deep hashing algorithms, we utilize the binary codes learned by GSDH\_P as labels to train classification models, by using the AlexNet architecture \cite{krizhevsky2012imagenet} in the Pytorch framework with pre-trained on the ImageNet database \cite{deng2009imagenet}, and name this method as GSDH\_P\(^{\ast}\). For SDH\_P and GSDH\_P, we empirically set \(p=1000\), \(n_{b}=100\), \(\beta=10\), \(L_1=20\) and \(L_2=3\). 

To evaluate the hashing methods, we utilize three major criteria: MAP, Precision and Recall, to evaluate their ranking performance on the single-label task, and employ two main criteria: NDCG \cite{jarvelin2000ir} and average cumulative gain (ACG) \cite{wang2013learning}, to assess their performance on multi-label tasks. Given a set of queries, MAP is the mean of the average precision (AP) for each query. AP is defined as:
\begin{equation}
AP@R = \frac{\sum_{k=1}^R P(k) \delta(k)}{\sum_{k=1}^{R}\delta(k)}
\end{equation}
where \(R\) is the number of top returned samples, \(P(k)\) is the precision at cut-off \(k\) in the list, \(\delta(k)=1\) if the sample ranked at \(k\)-th position is relevant, otherwise, \(\delta(k)=0\).

NDCG is the normalization of the discounted cumulative gain (DCG), which is calculated by \cite{jarvelin2000ir} :
\begin{equation}
DCG=rel_{1} + \sum_{k=2}^R \frac{rel_{k}}{log_{2}(k)}, 
\end{equation}
where \(k\) is the ranking position and \(rel_{k}\) is the relevance between the \(k\)-th retrieved sample and the query.

ACG denotes the average cumulative gain, it is defined as:
\begin{equation}
ACG_{r}=\frac{1}{\left | \mathcal{N}_r \right |} \sum_{x\in \mathcal{N}_r} rel_{x}
\end{equation}
where \(\mathcal{N}_r\) represents the retrieved samples within Hamming radius \(r\) and \(rel_{x}\) is the relevance between a returned sample \(\mathbf{x}\) and the query.

Since all non-deep hashing methods have similar test time, we only show their training time for better comparison. We utilize MATLAB and conduct all experiments on a 3.50GHz Intel Xeon E5-1650 CPU with 128GB memory.

\begin{table}[!htb]
\scriptsize 
\centering
\caption{Ranking performance (MAP) on top 500 retrieved samples and training time (seconds) of different hashing methods on the single-label database CIFAR-10.}
\vspace{0.2em}
\begin{tabular}{|c|c|c|c|c|c|}
\hline
\multirow{3}*{Method} & \multicolumn{5}{c|}{ GIST (\(n=10000\))}  \\ 
 \cline{2-6}
 & \multicolumn{4}{c|}{MAP (Top 500) } & Time  \\      
 \cline{2-6}
 & 8-bit&16-bit & 32-bit  & 64-bit & 64-bit  \\
\hline
KSH  & 0.4281  & 0.4706  & 0.5098  & 0.5270 &\(4.6\times10^{3}\)      \\
\hline
CCA-ITQ	 & 0.2014  &0.1984 &0.2147 &0.2316 &0.9   \\
\hline
SDH  & 0.4970 & 0.5370 &0.5670 &0.5781 &7.3     \\
\hline
COSDISH  & 0.5116 & \underline{0.5912} & \underline{0.5980} &\underline{0.6162} &\(1.5\times10^{1}\)          \\
\hline
KSDH  & \underline{0.5370} &0.5740 &0.5900 &0.6000 &3.5            \\
\hline
ADGH  &0.5360  &0.5700  &\underline{0.5980} &0.6020 &1.4              \\
\hline
\hline
RSH &0.2121 &0.1889  &0.1812 &0.1810 &\(7.5\times10^{4}\)           \\
\hline
CGH  & 0.3013 &0.3040 &0.3255  &0.3305 &\(9.2\times10^{2}\)       \\
\hline
Top-RSBC &0.2568 &0.2404 &0.2400 &0.2544 &\(9.3\times10^{4}\)     \\
\hline
DSeRH	 &0.2020  &0.2070 &0.2087 &0.2116 &\(2.7\times10^{3}\)   \\

\hline
\hline 
\textbf{SDH\_P} & \(\mathbf{0.5755}\)    & 0.5981   &0.6102 &0.6222   &\(1.1\times 10^{1}\)     \\
\hline
\textbf{GSDH\_P} &0.5600  & \(\mathbf{0.5992}\) & \(\mathbf{0.6272}\) &\(\mathbf{0.6300}\) & \(2.4\times 10^{1}\) \\
\hline
\hline
\multirow{3}*{Method} & \multicolumn{5}{c|}{GIST (Full)}  \\ 
 \cline{2-6}
 & \multicolumn{4}{c|}{MAP (Top 500) } & Time  \\      
 \cline{2-6}
 & 8-bit&16-bit & 32-bit  & 64-bit & 64-bit  \\
 \hline
KSH &0.4057 &0.4725  &0.5126& 0.5317  & \(3.5\times10^{4}\) \\
\hline
CCA-ITQ	 &0.2338 &0.2224 &0.2473  &0.2745   & \(4.7\)    \\
\hline
SDH  & 0.4723 &0.5700 &0.5920&0.6038   & \(5.2\times 10^{1}\)    \\
\hline
COSDISH  & 0.5593  &0.6065  &\underline{0.6125}&\underline{0.6255} & \(9.2\times 10^{1}\)         \\
\hline  
KSDH  &0.5687  &0.5955   &0.6015&0.6091   & \(7.6\times 10^{1}\)             \\
\hline
ADGH  &\underline{0.5731}&\underline{0.6097} &0.6113   &0.6119          & \(4.4\)             \\
\hline
\hline
\textbf{SDH\_P} &\(\mathbf{0.5735}\) &\(\mathbf{0.6172}\) &0.6222 &0.6279   & \(3.3\times 10^{1}\)        \\
\hline
\textbf{GSDH\_P} &0.5680 &0.6142 & \(\mathbf{0.6239}\) & \(\mathbf{0.6333}\) & \(6.9\times 10^{1}\)      \\
\hline
\end{tabular}
\label{cifar10}
\vspace{-1em}
\end{table}

\begin{table}[!htb]
\scriptsize 
\centering
\caption{Ranking performance (MAP) on top 1000 retrieved samples and training time (seconds) of KSH, BRE, MLH and their variants on the CIFAR-10 database.}
\vspace{0.2em}
\begin{tabular}{|c|c|c|c|c|c|}
\hline
\multirow{3}*{Method} & \multicolumn{5}{c|}{ GIST (\(n=10000\))}  \\ 
 \cline{2-6}
 & \multicolumn{4}{c|}{MAP (Top 1000) } & Time  \\      
 \cline{2-6}
 & 8-bit&16-bit & 32-bit  & 64-bit & 64-bit  \\
\hline
KSH  & \underline{0.4069}  &\underline{0.4458} &\underline{0.4908}   &\underline{0.5121}  &\(4.6\times10^{3}\)      \\
\hline
BRE	 & 0.1999  &0.2028   &0.2187 &0.2340 & \(1.3\times 10^{5}\)   \\
\hline
MLH  &0.1509   &0.2581   &0.2491 &0.2467 & \(2.5\times 10^{2}\)    \\
\hline
\textbf{GSDH\_P}\(_{KSH}\)  &0.5528 &\(\mathbf{0.5996}\) &0.6126 &0.6162 & \(2.4\times 10^{1}\)        \\
\hline
\textbf{GSDH\_P}\(_{BRE}\)  &0.5478 &0.5960 &\(\mathbf{0.6158}\) & \(\mathbf{0.6251}\) &  \(5.8\times 10^{1}\)             \\
\hline
\textbf{GSDH\_P}\(_{Hinge}\) &\(\mathbf{0.5582}\) &0.5961 &0.6136 &0.6240 & \(1.2\times 10^{2}\)          \\
\hline
\hline
FastH\(_{KSH}\) &\underline{0.5364} &\underline{0.5814} &0.6106 &0.6207 & \(4.1\times 10^{2}\)            \\
\hline
FastH\(_{BRE}\) &0.5200 &0.5800 &0.6104 &0.6139 &  \(4.1\times 10^{2}\)      \\
\hline
FastH\(_{Hinge}\) &0.5243 &0.5775 &\underline{0.6150} &\underline{0.6305} &\(3.7\times 10^{2}\)      \\
\hline
\textbf{GSDH\_P}\(_{KSH}\)+DT  & \(\mathbf{0.5482}\) &0.5921 &0.6207 &0.6282 &\(1.7\times 10^{2}\)         \\
\hline
\textbf{GSDH\_P}\(_{BRE}\)+DT  &0.5221&0.5862 &\(\mathbf{0.6224}\) &0.6292 & \(1.8\times 10^{2}\)               \\
\hline
\textbf{GSDH\_P}\(_{Hinge}\)+DT &0.5352 & \(\mathbf{0.5972}\) &0.6215 & \(\mathbf{0.6340}\) & \(2.1\times 10^{2}\)           \\
\hline
\end{tabular}
\label{cifar10_dt}
\vspace{-1em}
\end{table}

%
%

\begin{table}[!htb]
\scriptsize 
\centering
\caption{Ranking performance (MAP) of GSDH\_P\(^{\ast}\) and several popular deep hashing algorithms on the CIFAR-10 database (\(^{\dagger}\) denotes that the shown results are reported in  \cite{cao2018deep}).}
\vspace{0.2em}
\begin{tabular}{|c|c|c|c|c|}
\hline
\multirow{2}*{Method} & \multicolumn{4}{c|}{MAP @ Top 5000 }  \\ 
 \cline{2-5}
 & 12-bit & 24-bit & 32-bit  & 48-bit\\
\hline
CNNH  &0.429   &0.511   &0.509   &0.522     \\
 \hline
 DNNH &0.552  &0.566 &0.558 &0.581    \\
 \hline
 DHN  &\underline{0.555} &\underline{0.594} &\underline{0.603} &\underline{0.621}          \\

\hline
GSDH\_P\(^{\ast}\)  & \(\mathbf{0.791}\) & \(\mathbf{0.786}\) & \(\mathbf{0.780}\) & \(\mathbf{0.776}\)            \\
\hline
\hline
\multirow{2}*{Method} & \multicolumn{4}{c|}{ MAP @ H\(\leq\) 2}  \\ 
 \cline{2-5}
 & 16-bit & 32-bit & 48-bit  & 64-bit\\
 \hline
 CNNH\(^{\dagger}\)  &0.5512 &0.5468 &0.5454   &0.5364     \\
 \hline
 DNNH\(^{\dagger}\) 	 &0.5703   &0.5985&0.6421 &0.6118    \\
 \hline
 DHN\(^{\dagger}\)   & 0.6929   &0.6445   &0.5835 &0.5883      \\
 \hline
 HashNet\(^{\dagger}\)  &0.7576 &0.7776 &0.6399 &0.6259          \\
 \hline
 DCH\(^{\dagger}\)  & \underline{0.7901} & \underline{0.7979} &\underline{0.8071} &\underline{0.7936}          \\
 \hline
 \textbf{GSDH\_P\(^{\ast}\)}  &\(\mathbf{0.8480}\)&\(\mathbf{0.8480}\)&\(\mathbf{0.8400}\)  &\(\mathbf{0.8329}\)            \\
\hline
\end{tabular}
\label{cifar10_dh}
\vspace{-1em}
\end{table}

\vspace{-0.5em}
\subsection{Experiments on CIFAR-10}

We partition the CIFAR-10 database into training and query sets, which consist of 50K and 10K images, respectively. Each image is aligned and cropped to \(32\times 32\) pixels and then represented by a 512-dimensional GIST feature vector \cite{oliva2001modeling}.  In our experiments, we kernelize GIST feature vectors by using the same kernel type in KSH \cite{WJR} and uniformly selecting 1K samples from the training set as anchors. Since some comparative multi-wise based algorithms are extremely slow when using a large number of training data, we utilize only a subset of data to train models for the non-deep hashing algorithms: KSH, CCA-ITQ, SDH, COSDISH, KSDH, ADGH, RSH, CGH, Top-RSBC and DSeRH. Similar to KSH, we uniformly pick up 1K and 100 images from each category for training and testing, respectively. Then, we evaluate the proposed algorithms and the six non-ranking algorithms (KSH, CCA-ITQ, SDH, COSDISH, KSDH and ADGH) using all training and query images. Their ranking performance in term of MAP with 500 retrieved samples is shown in Table \ref{cifar10}. As we can see, when using 10K training and 1K query images, both SDH\_P and GSDH\_P achieve better MAPs than the other algorithms at 8-, 16-, 32- and 64-bit. The gain of GSDH\_P ranges from 1.35\% to 4.88\% over the best competitor except SDH\_P.  Additionally, GSDH\_P obtains higher MAPs than SDH\_P at 16-, 32- and 64-bit. When using 50K training and 10K query images, GSDH\_P outperforms the other comparative non-deep hashing algorithms except ADCG at 8-bit. Figure \ref{figprecision} presents the precision and recall of various hashing algorithms at 8-, 16-, 32-, 64- and 128-bits on the CIFAR-10 database with Hamming radius being 2. It further demonstrates the superior performance of GSDH\_P over the other hashing algorithms. Although SDH\_P achieves best precision at 16- and 32-bit, its recall is very low compared to other algorithms.

To illustrate the generation of the proposed algorithm GSDH\_P, we uniformly pick up 1K and 100 images from each category for training and testing, respectively. We repeat this process 10 times and report the average MAP of KSH, BRE, MLH, FastH and GSDH\_P with top 1K samples returned in Table \ref{cifar10_dt}. Note that GSDH\_P\(_{KSH}\)+DT denotes learning classification models by using decision trees as classifiers and the learned binary codes of GSDH\_P\(_{KSH}\) as labels. Similar definitions for GSDH\_P\(_{BRE}\)+DT and GSDH\_P\(_{Hinge}\)+DT. Here, we utilize GSDH\_P\(_{KSH}\) to represent GSDH\_P for a clear comparison. Table \ref{cifar10_dt} illustrates that GSDH\_P can achieve significantly better performance than KSH, BRE and MLH. Meanwhile, it also outperforms FastH with lower training costs.

When evaluating the deep hashing algorithms, we follow the experimental protocol in \cite{xia2014supervised}, i.e. randomly selecting 500 and 100 images per class for training and testing, respectively. We show the reported MAP of CNNH, DNNH and DHN with 5K samples retrieved and the MAP of CNNH, DNNH, DHN, HashNet and DCH with Hamming radius being 2 in Table \ref{cifar10_dh}, which shows that GSDH\_P\(^\ast\) significantly outperforms recent state of the arts on the CIFAR-10 database.

\renewcommand{\thesubfigure}{\relax}
\begin{figure}[tb]
\center
\subfigure[]{\includegraphics[trim={2em 3em 2.5em 2em}, clip, width=0.07\textwidth]{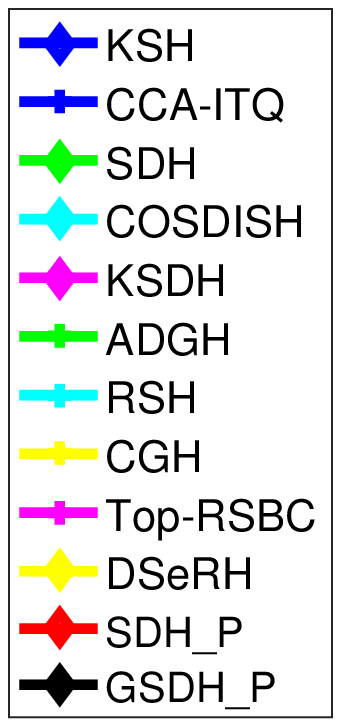}}
\setcounter{subfigure}{0}
\subfigure[(a) Precision ]{\includegraphics[trim={1em 0em 2em 1em}, clip, width=0.19\textwidth]{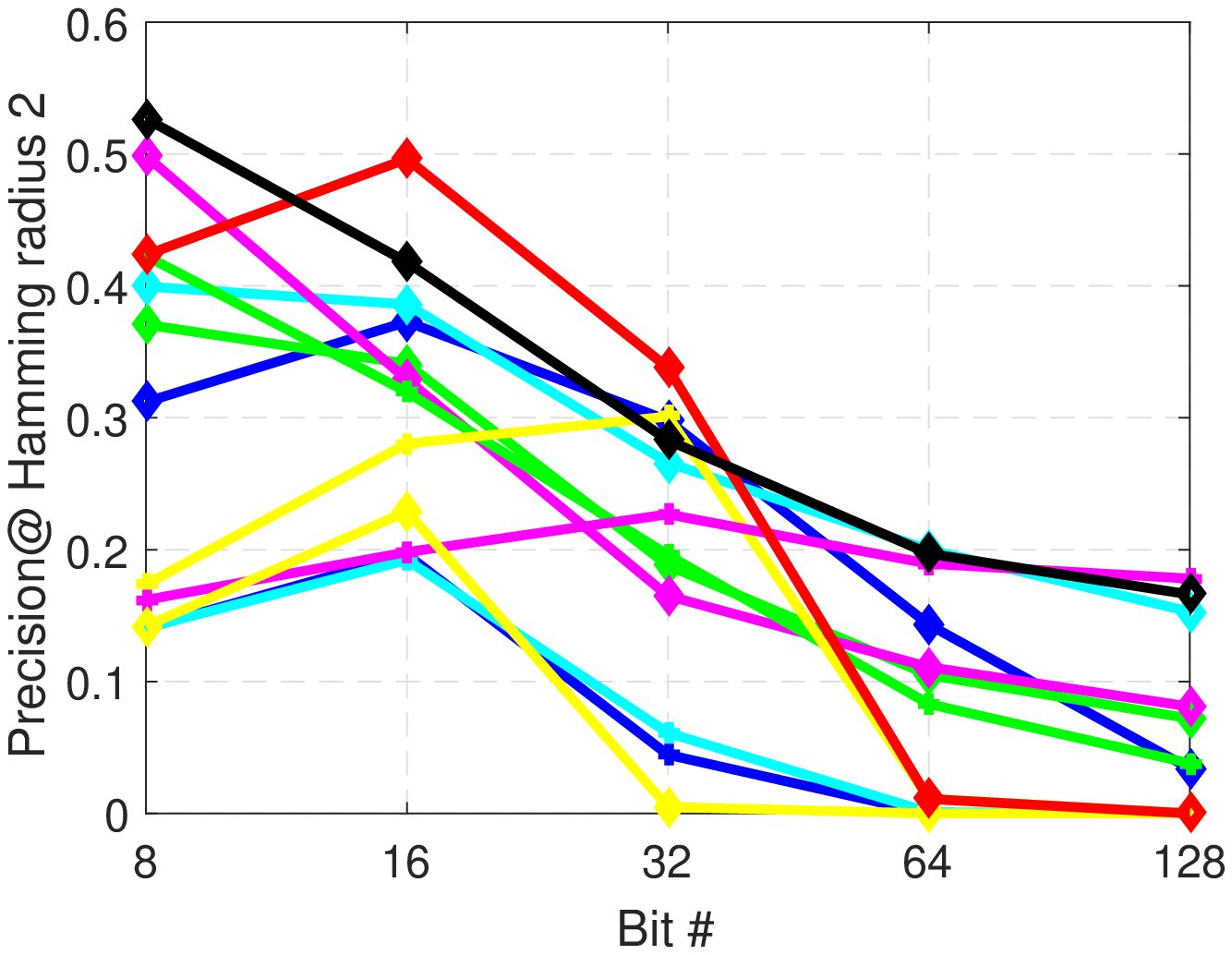}}
\subfigure[(b) Recall]{\includegraphics[trim={1em 0em 2em 1em}, clip, width=0.19\textwidth]{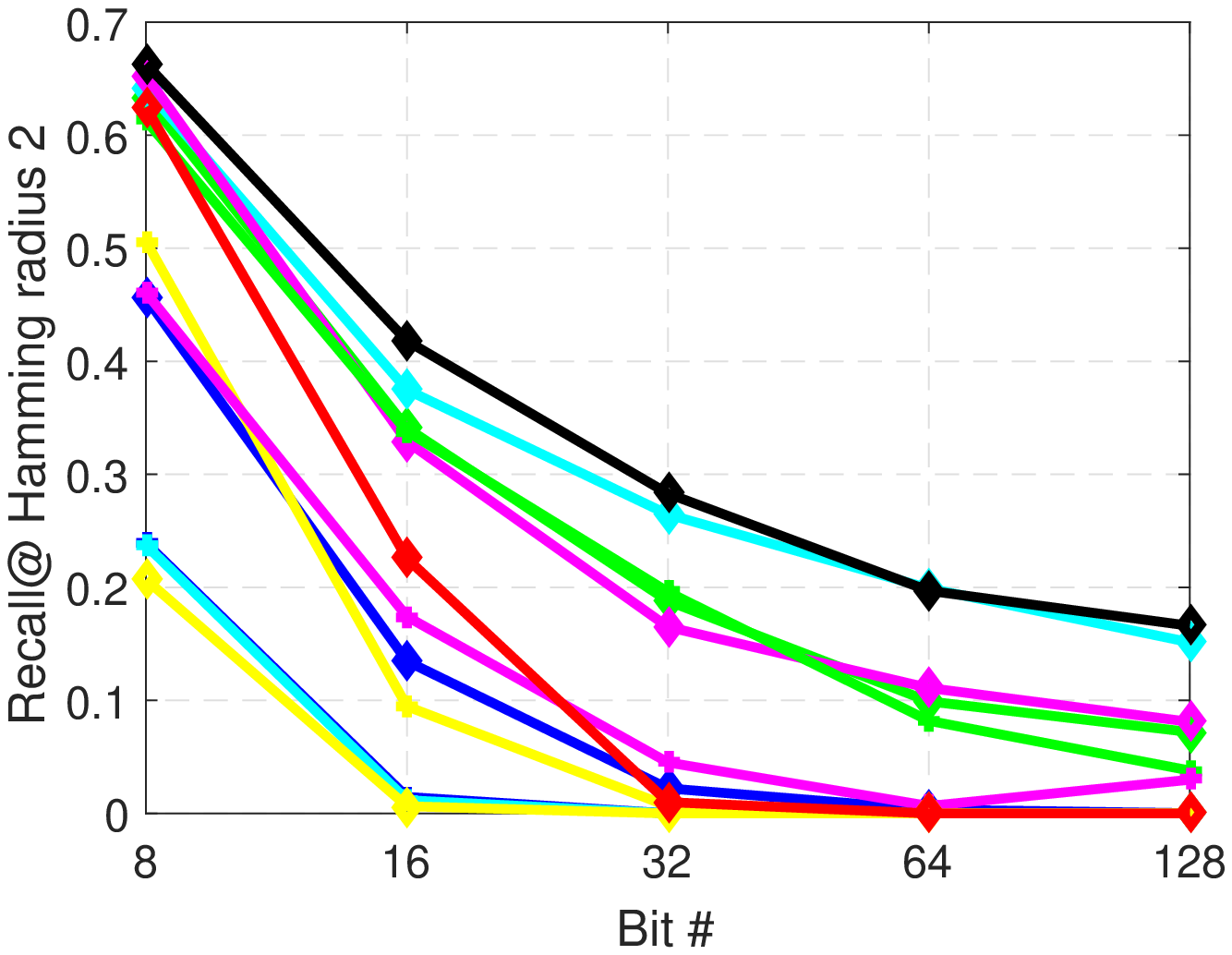}}
\vspace{-1em}
\caption{Precision and Recall vs. Bits of various algorithms on the CIFAR-10 database with Hamming radius being 2.}
\label{figprecision}
\vspace{-1em}
\end{figure}

\begin{table}[htb]
\scriptsize 
\centering
\caption{Ranking performance (NDCG) on top 50 retrieved samples and training time (seconds) of different hashing algorithms on the multi-label database NUS-WIDE.}
\vspace{0.2em}
\begin{tabular}{|c|c|c|c|c|c|}
\hline
\multirow{3}*{Method} & \multicolumn{5}{c|}{ BoW (\(n=10000\))}  \\ 
 \cline{2-6}
 & \multicolumn{4}{c|}{NDCG (Top 50) } & Time  \\      
 \cline{2-6}
 & 8-bit&16-bit & 32-bit  & 64-bit & 64-bit  \\
\hline
KSH   & 0.2755  &0.2730 & 0.2554 &0.2395  &\(2.1\times 10^{3}\)    \\
\hline
CCA-ITQ	 &0.2522 &0.2337 &0.2248 &0.2145  & 0.4   \\
\hline
SDH  &0.3108 & \underline{0.3674} & 0.3674 &0.3494 &\(2.5\times 10^{1}\)    \\
\hline
COSDISH  & 0.3169  & 0.3606 & \underline{0.3989} & \underline{0.3801}  & \(3.5\times 10^{1}\)          \\
\hline
KSDH  & \underline{0.3227} & 0.3125 & 0.3091 & 0.3092   & 5.3                \\
\hline
\hline
RSH 	&0.2040  &0.2061 &0.2088 &0.2193 & \(9.2\times 10^{4}\)           \\
\hline
CGH   &0.2163  & 0.2228    &0.2475  & 0.2323   & \(1.5\times 10^{3}\)       \\
\hline
Top-RSBC  & 0.1962 &0.2238     &0.2564  & 0.2758      & \(8.9\times 10^{4}\)     \\
\hline
DSeRH	  & 0.2158 &0.2264 &0.2271  &0.2210 &\(8.2\times 10^{3}\)    \\
\hline
\hline 
\textbf{SDH\_P} &0.2916  & 0.3014 & 0.3038  & 0.3174  & \(4.3\)     \\
\hline
\textbf{GSDH\_P} & \(\mathbf{0.3588}\)&\(\mathbf{0.3889}\)& \(\mathbf{0.4083}\) &\(\mathbf{0.4311}\)  & \(5.5\times 10^{1}\)        \\
\hline 
\hline
\multirow{3}*{Method} & \multicolumn{5}{c|}{BoW (Full)}  \\ 
 \cline{2-6}
 & \multicolumn{4}{c|}{NDCG (Top 50) } & Time  \\      
 \cline{2-6}
 & 8-bit&16-bit & 32-bit  & 64-bit & 64-bit  \\
 \hline
CCA-ITQ	&0.2562 &0.2551 &0.2528 &0.2475  &8.8   \\
\hline
SDH & 0.3381  & \underline{0.3927} &0.3838 &0.3827 &\(3.8\times 10^{2}\)  \\
\hline
COSDISH   & \underline{0.3649} & 0.3539 & \underline{0.4279} &\underline{0.4323} &\(3.5\times 10^{2}\)   \\
\hline
\hline
\textbf{SDH\_P}  &0.3010  &0.3094& 0.3431 &0.3313 &\(7.7\times 10^{1}\)             \\
\hline
\textbf{GSDH\_P} & \(\mathbf{0.4152}\) & \(\mathbf{0.4521}\)&\(\mathbf{0.4646}\) &\(\mathbf{0.4875}\)  & \(5.0\times 10^2\)       \\
\hline
\end{tabular}
\label{nuswide}
\vspace{-1em}
\end{table}

\begin{table}[!htb]
\scriptsize 
\centering
\caption{Ranking performance (MAP) of GSDH\_P\(^{\ast}\) and several popular deep hashing algorithms on the NUSWIDE database (\(^{\dagger}\) denotes that the shown results are reported in  \cite{cao2018deep}).}
\vspace{0.2em}
\begin{tabular}{|c|c|c|c|c|}
\hline
\multirow{2}*{Method} & \multicolumn{4}{c|}{MAP @ Top 5000 }  \\ 
 \cline{2-5}
 & 12-bit & 24-bit & 32-bit  & 48-bit\\
\hline
CNNH  &0.617 &0.663 &0.657   &0.688     \\
\hline
DNNH	& 0.674   &0.697 &0.713 &0.715    \\
\hline
DHN  & \underline{0.708} & \underline{0.735} &\underline{0.748} &\underline{0.758}          \\
\hline
GSDH\_P\(^{\ast}\)  & \(\mathbf{0.759}\) &\(\mathbf{0.791}\) &\(\mathbf{0.806}\) &\(\mathbf{0.799}\)            \\
\hline
\hline
\multirow{2}*{Method} & \multicolumn{4}{c|}{ MAP @ H\(\leq\) 2}  \\ 
 \cline{2-5}
 & 16-bit & 32-bit & 48-bit  & 64-bit\\
 \hline
 CNNH\(^{\dagger}\)  & 0.5843  &0.5989   &0.5734   &0.5729     \\
 \hline
 DNNH\(^{\dagger}\)	& 0.6191   &0.6216 &0.5902 &0.5626    \\
 \hline
 DHN\(^{\dagger}\)  &0.6901   &0.7021   &0.6736 &0.6190      \\
 \hline
 HashNet\(^{\dagger}\) &0.6944 &0.7147 &0.6736 &0.6190          \\
 \hline
 DCH\(^{\dagger}\) & \underline{0.7401} &\underline{0.7720} &\underline{0.7685} &\underline{0.7124}          \\
 \hline
 \textbf{GSDH\_P\(^{\ast}\)}  & \(\mathbf{0.8073}\)& \(\mathbf{0.8025}\) & 0.7572 & \(\mathbf{0.7442}\)             \\
\hline
\end{tabular}
\label{nusiwde_dh}
\vspace{-1em}
\end{table}

\renewcommand{\thesubfigure}{\relax}
\begin{figure}
\center
\subfigure[]{\includegraphics[trim={18em 7em 18.5em 7em}, clip, width=0.07\textwidth]{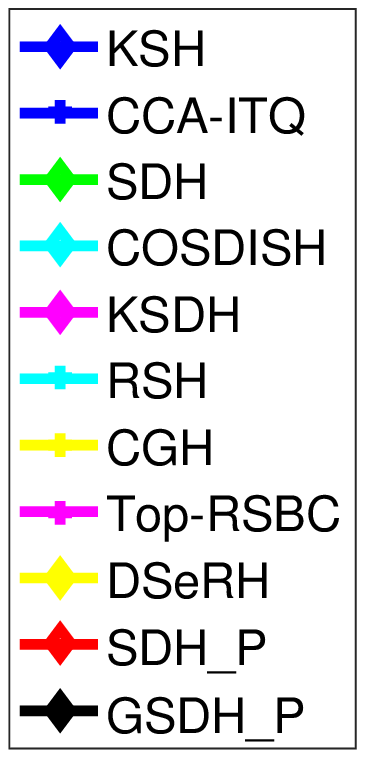}}
\setcounter{subfigure}{0}
\subfigure[(a) ACG ]{\includegraphics[trim={1em 0em 2em 1em}, clip, width=0.19\textwidth]{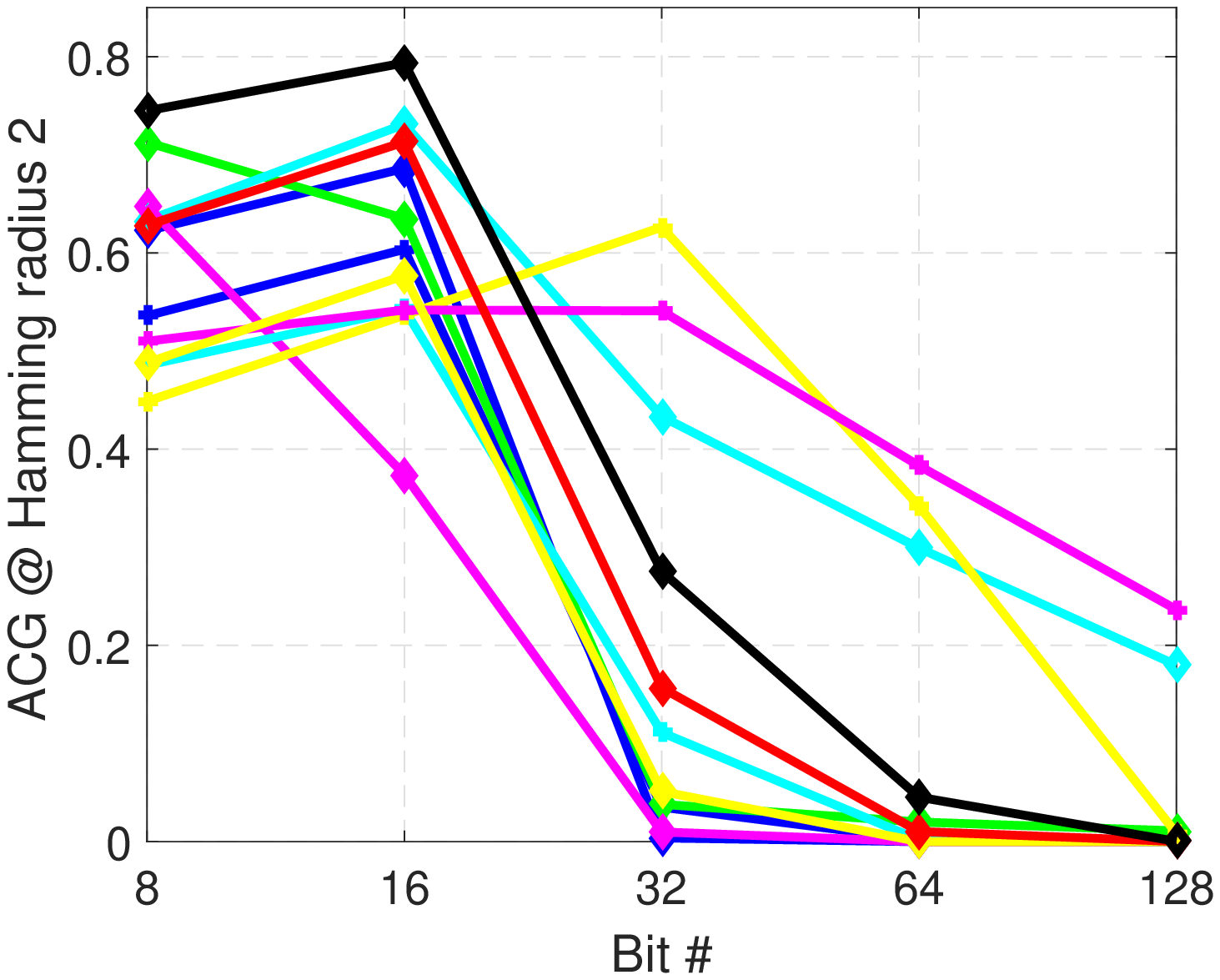}}
\subfigure[(b) NDCG]{\includegraphics[trim={1em 0em 2em 1em}, clip, width=0.19\textwidth]{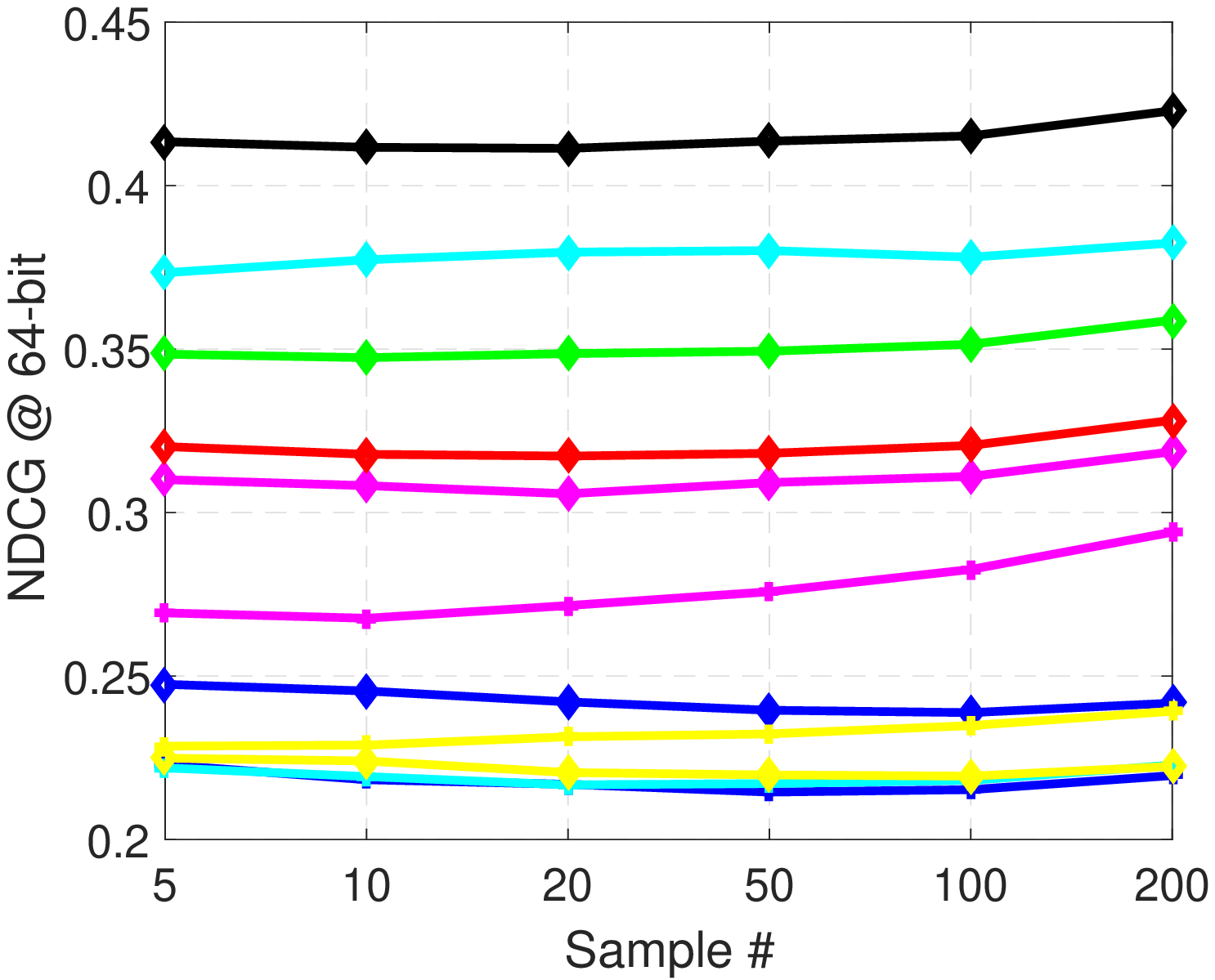}}
\vspace{-1em}
\caption{ACG vs. Bits and NDCG vs. retrieved samples of various algorithms with images selected from the NUS-WIDE database.}
\label{figndcg}
\vspace{-1.5em}
\end{figure}

\vspace{-0.5em}
\subsection{Experiments on NUS-WIDE}

For the NUS-WIDE database, we partition all images into training and test sets, including around 185K training and 10K query images, respectively. Each image is represented by the provided 500 Bag-of-Words (BoW) features.  In our experiments, similar to KSH \cite{WJR}, we kernelize BoW feature vectors by uniformly selecting 1K samples from the training set as anchors. Firstly, we uniformly select 10K training and 1K query images to evaluate all ranking and non-ranking algorithms except ADGH, since it cannot be directly applied to tackling multi-label tasks. Then, we utilize all training and query images to evaluate the proposed algorithms and three scalable algorithms: CCA-ITQ, SDH and COSDISH. Table \ref{nuswide} presents their ranking performance in term of NDCG with 50 samples retrieved. It shows that GSDH\_P has superior ranking performance to the other nine non-deep hashing algorithms when 10K training images are used. Its gain in term of NDCG is from 2.35\% to 13.42\% over the best competitors except SDH\_P. When all training images are used, GSDH\_P significantly outperforms CCA-ITQ, SDH and COSDISH, its gain ranges from 8.58 \% to 15.13\% over the best competitors at all four bits. Figure \ref{figndcg} presents the ACG of various algorithms on the NUS-WIDE database at 8-, 16-, 32-, 64- and 128-bit with Hamming radius being 2, and their NDCGs with 5, 10, 20, 50, 100 and 200 samples retrieved at 64-bit. It further illustrates that GSDH\_P outperforms the other hashing algorithms.

Additionally, we follow the protocol in \cite{xia2014supervised} \cite{lai2015simultaneous} to randomly select 100 query and uniformly sample 500 training images from each of the selected 21 most frequent labels, in order to evaluate the MAP of GSDH\_P\(^{\ast}\) and deep hashing algorithms CNNH, DNNH and DHN with top 5000 images returned. Moreover, to evaluate the MAP of GSDH\_P\(^{\ast}\), CNNH, DNNH, DHN, HashNet and DCH with Hamming radius being 2, we follow the experimental protocol in DCH \cite{cao2018deep}, i.e. randomly sample 5K and 10K images to construct testing and training sets, respectively. Note that when calculating MAP, if two images share at least one label, they are similar and \(s_{ij}=1\), otherwise, they are dissimilar and \(s_{ij}=0\). Table \ref{nusiwde_dh} shows their MAP with top 5K retrieved samples and Hamming radius being 2. It figures out that GSDH\_P\(^{\ast}\) can significantly outperform CNNH, DNNH and DHN when top 5K images are returned, and when Hamming radius being 2, GSDH\_P\(^{\ast}\) achieves 6.72\%, 3.05\% and 3.18\% higher MAP than the best competitor DCH at 16-, 32- and 64-bit.  

\begin{table}[t]
\scriptsize 
\centering
\caption{Ranking performance (NDCG) on top 50 retrieved samples and training time (seconds) of various hashing algorithms on images from the multi-label database COCO.}
\begin{tabular}{|c|c|c|c|c|c|}
\hline
\multirow{3}*{Method}& \multicolumn{5}{c|}{n=10000}   \\
\cline{2-6}
& \multicolumn{4}{c|}{NDCG (Top 50)} & Time    \\ 
 \cline{2-6}
& 8-bit&16-bit & 32-bit &64-bit & 64-bit \\
\hline
KSH  & \underline{0.2352}  &0.3226   &0.3810    &0.4168  & \(2.7\times 10^{3}\)   \\
\hline
CCA-ITQ  &0.2293   &\underline{0.3367} &\underline{0.3947}  &0.4103     & 0.3      \\
\hline
SDH 	 & 0.1911  &0.3218   &0.3911 &\underline{0.4345}    &  \(2.1\times 10^{1}\)    \\
\hline
COSDISH   & 0.1623 &0.2382   &0.2438 &0.2890      & \(2.2\times 10^{1}\)      \\
\hline
KSDH  &0.1623 &0.2000   &0.2292    &0.2597      & 2.7   \\
\hline
\hline
RSH 	&0.2156  &0.2751   &0.3299    &0.3583  & \(6.1\times 10^{4}\)          \\
\hline
CGH   & 0.1566  &0.1928 &0.2176 &0.2324  & \(1.9\times 10^{3}\)      \\
\hline
Top-RSBC  &0.1547  &0.1843   &0.2045    &0.2845 & \(8.6\times 10^{4}\)    \\
\hline
DSeRH	  &0.2144  &0.2556 &0.3210 &0.3560  & \(1.7\times 10^{3}\)    \\
\hline
\hline
\textbf{SDH\_P} & 0.1932   & 0.2476  &0.3165    &0.3465   & 2.7    \\
\hline
\textbf{GSDH\_P} & \(\mathbf{0.2526}\) & \(\mathbf{0.3550}\)   & \(\mathbf{0.4136}\)    & \(\mathbf{0.4385}\)    & \(5.5\times 10^{1}\) \\
\hline 
\hline
\multirow{3}*{Method} & \multicolumn{5}{c|}{Full}  \\ 
 \cline{2-6}
& \multicolumn{4}{c|}{NDCG (Top 50) } & Time  \\      
 \cline{2-6}
& 8-bit &16-bit & 32-bit  & 64-bit & 64-bit  \\
 \hline
CCA-ITQ	& \underline{0.1831} &0.2426  &0.2902  &0.3032   & \(1.0\times 10^{1}\)   \\
\hline
SDH &0.1684  &\underline{0.2430} &\underline{0.3537}  &\underline{0.3846}  & \(3.2\times 10^2\) \\
\hline
COSDISH   &0.1557 &0.1741 &0.2013 &0.2146  & \(1.9\times 10^{2}\)   \\
\hline
\hline
\textbf{SDH\_P}  &0.1756  &0.2427  &0.3252  &0.3471 & \(2.3\times 10^{1}\)            \\
\hline
\textbf{GSDH\_P} & \(\mathbf{0.2240}\)  & \(\mathbf{0.3160}\)  & \(\mathbf{0.3760}\)  & \(\mathbf{0.4020}\)  & \(6.1\times 10^2\)     \\
\hline
\end{tabular}
\label{coco}
\vspace{-1em}
\end{table}

\vspace{-0.5em}
\subsection{Experiments on COCO}
For the COCO database, we adopt 83K training images and select 10K validation images to construct training and query sets, respectively. Then we resize each image to \(32\times 32\) pixels and represent each one by using a 2048-dimensional CNN feature vector, which is extracted by a popular and powerful neural network: ResNet50 \cite{he2016deep}. After that, we kernelize the features by using the same kernel type as KSH with 1000 anchors selected. Firstly, we uniformly select 10K training and 1K query images to evaluate the proposed algorithms SDH\_P and GSDH\_P, and nine non-deep hashing algorithms. Then we adopt all training and query images to evaluate the scalable hashing algorithms CCA-ITQ, SDH, COSDISH, SDH\_P and GSDH\_P.


Table \ref{coco} displays their ranking performance in term of NDCG with top 50 retrieved samples and training time of various algorithms on the COCO database. It illustrates that GSDH\_P has the superior performance to the other hashing algorithms. When using 10K training images, the gain of GSDH\_P in term of NDCG is 7.40\%, 5.44\%, 4.79\% and 0.92\% compared to the best competitor except SDH\_P at 8-, 16-, 32- and 64-bit, respectively; when using all training images, the NDCG of GSDH\_P is 4.09\%, 7.30\%, 2.23\% and 1.74\% higher than the best competitor except SDH\_P at the four bits, respectively. Figure \ref{figncocondcg} presents the ACG of various algorithms on the COCO database at 8-, 16-, 32-, 64- and 128-bit with Hamming radius being 2, and presents their NDCGs with 5, 10, 20, 50, 100 and 200 samples retrieved at 64-bit. It further demonstrates the strength of GSDH\_P.

\renewcommand{\thesubfigure}{\relax}
\begin{figure}
\center
\subfigure[]{\includegraphics[trim={18em 7em 18.5em 7em}, clip, width=0.07\textwidth]{figures/legend.eps}}
\setcounter{subfigure}{0}
\subfigure[(a) COCO ]{\includegraphics[trim={1em 0em 2em 1em}, clip, width=0.19\textwidth]{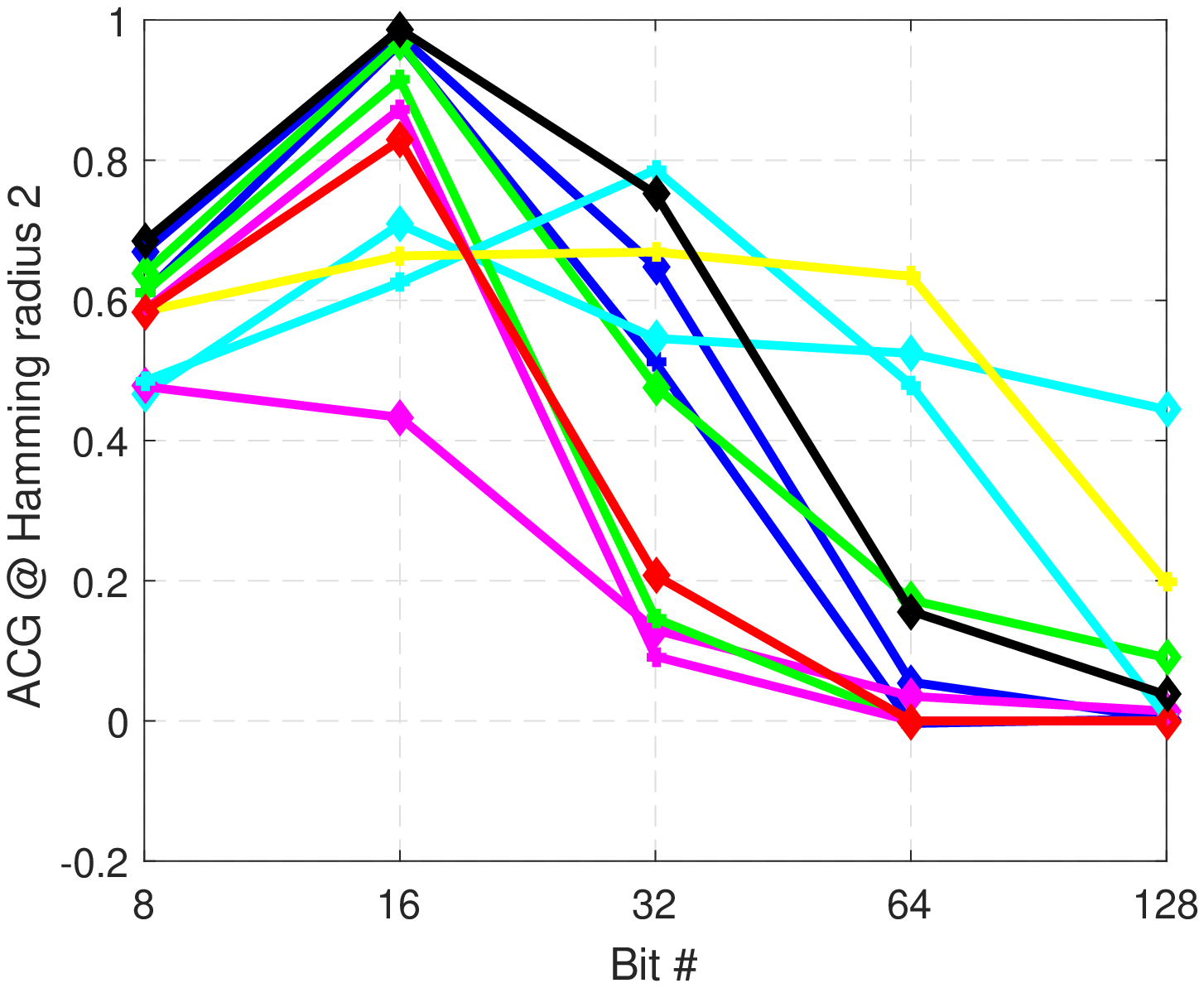}}
\subfigure[(b) COCO]{\includegraphics[trim={1em 0em 2em 1em}, clip, width=0.19\textwidth]{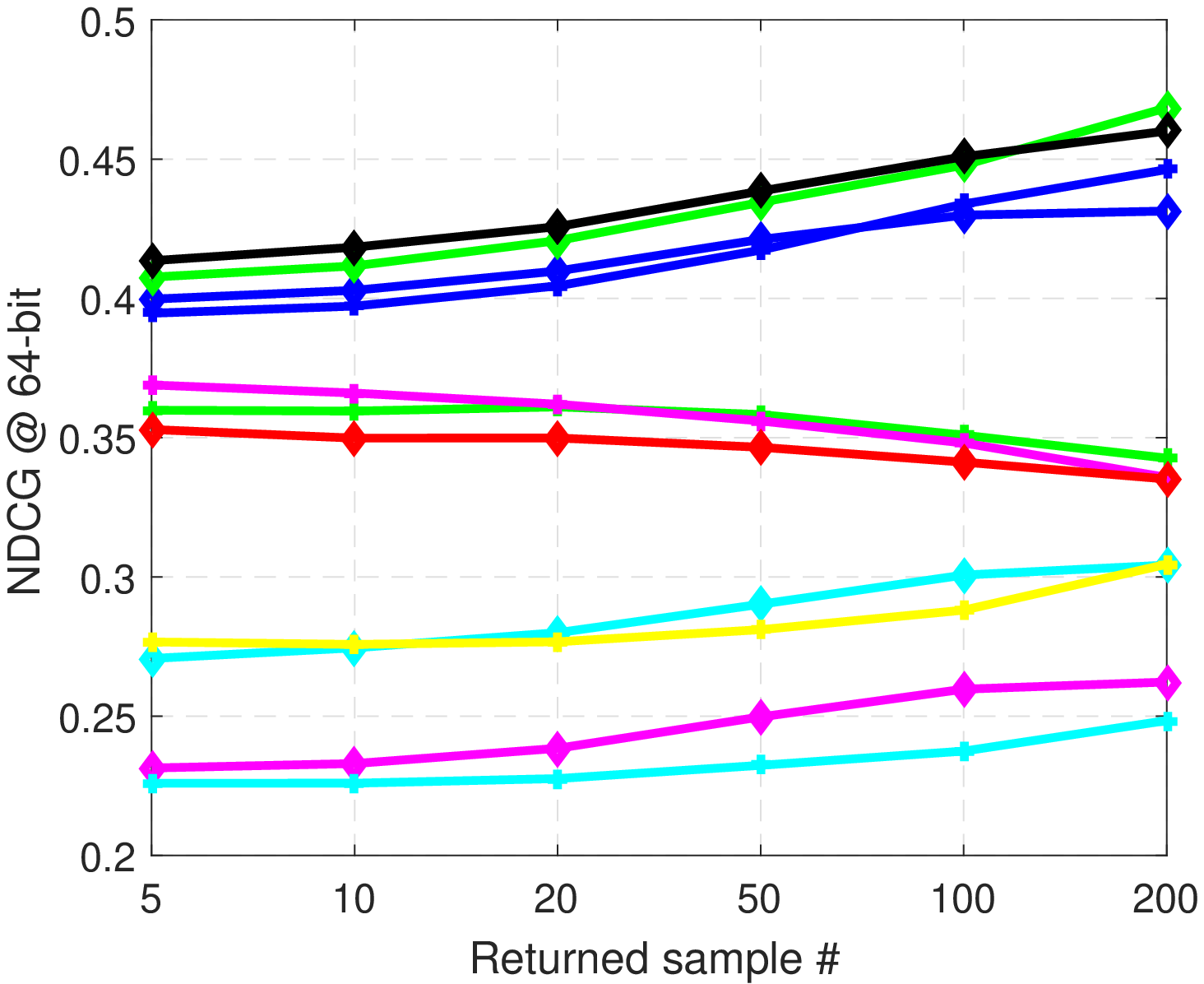}}
\vspace{-1em}
\caption{ACG vs. Bits and NDCG vs. retrieved samples of various algorithms with images selected from the NUS-WIDE database.}
\label{figncocondcg}
\vspace{-1.5em}
\end{figure}

\vspace{-0.5em}
\subsection{Parameter Influence}
Here, we mainly evaluate two essential parameters: batch size \(n_{b}\) and regularization parameter \(\beta\), where \(n_{b}\) and \(\beta\) determine the batch number \(f\) and the parameter \(\gamma\), respectively. Similar to previous experiments, we uniformly select 10K training and 1K query images from CIFAR-10, NUS-WIDE and COCO databases, and then encode each image features into 16-bit binary codes. Figure \ref{fignpara} shows the influence of \( n_{b}\in \left [1,10,10^2,10^3,5\times 10^3,10^4 \right]\) and \(\beta \in \left [ 0,1,10,10^2,10^3,10^4,10^5 \right ]\) on ranking performance in term of MAP and NDCG on the three databases. It suggests that both SDH\_P and GSDH\_P can obtain the best or sub-best ranking performance when \(n_{b} \in \left [ 1, 100 \right]\) and \(\beta \in \left [ 0, 100 \right]\) on all the three databases. Similar results can be found at other bits, we do not show them for brevity. In our single-label and multi-label experiments, without loss of generality, we empirically set \(\beta=10\) and \(n_{b}=100\) for the proposed algorithms.

\renewcommand{\thesubfigure}{\relax}
\begin{figure*}
\center
\small
\subfigure[(a) CIFAR@ SDH\_P]{\includegraphics[trim={1em 1em 1em 1em}, clip, width=0.16\textwidth]{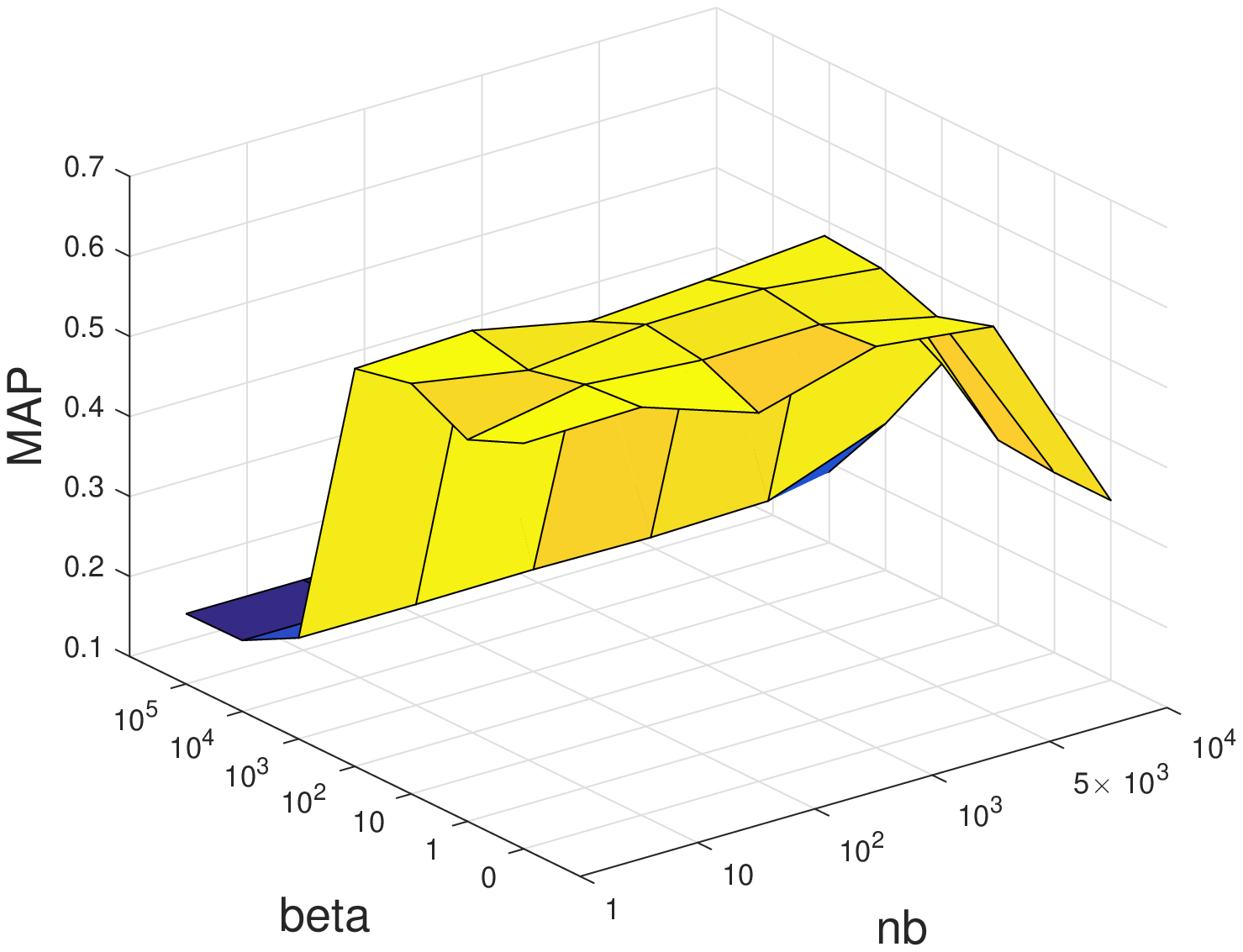}}
\subfigure[(b) CIFAR@ GSDH\_P]{\includegraphics[trim={1em 1em 1em 1em}, clip, width=0.16\textwidth]{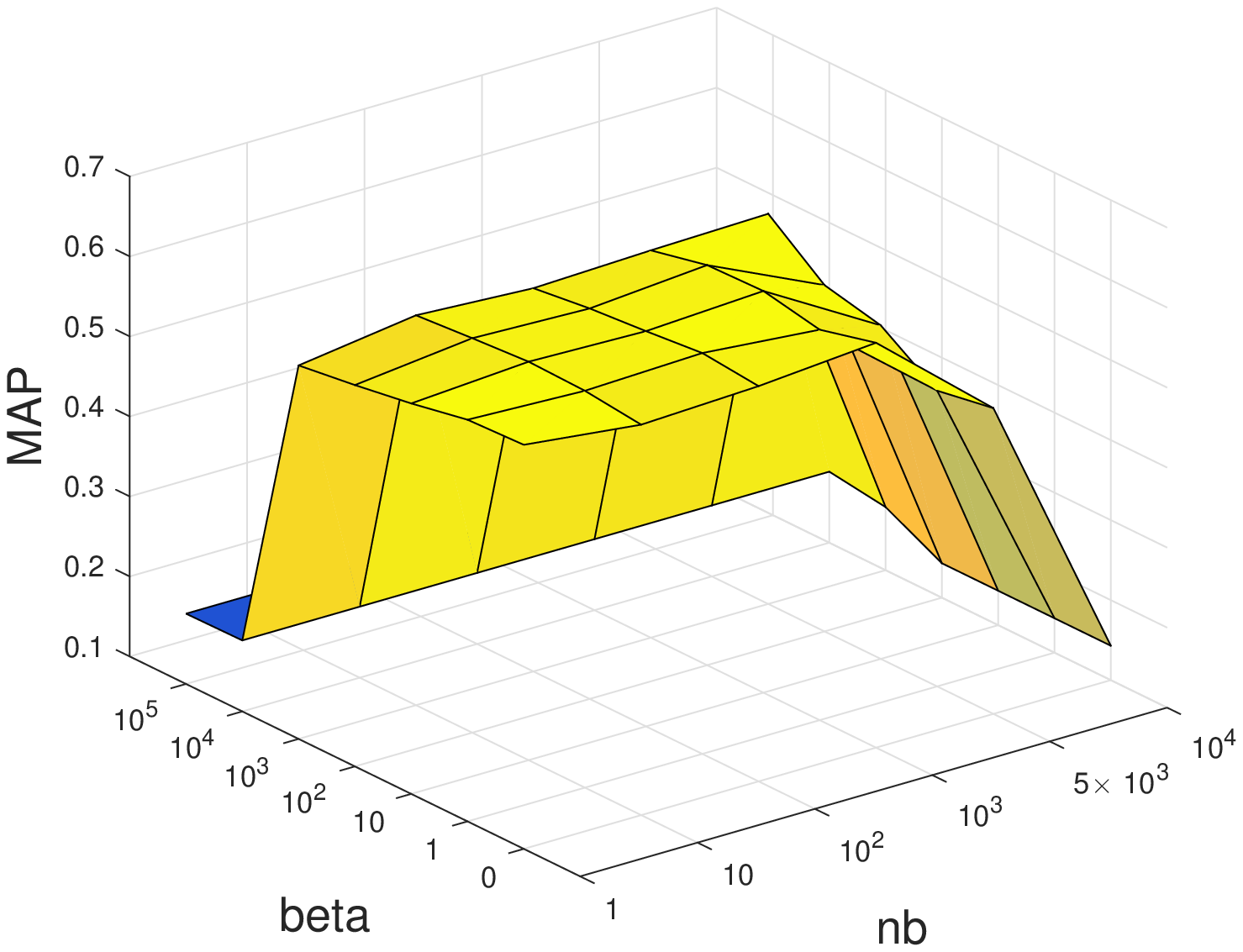}}
\subfigure[(c) NUS@ SDH\_P]{\includegraphics[trim={1em 1em 1em 1em}, clip, width=0.16\textwidth]{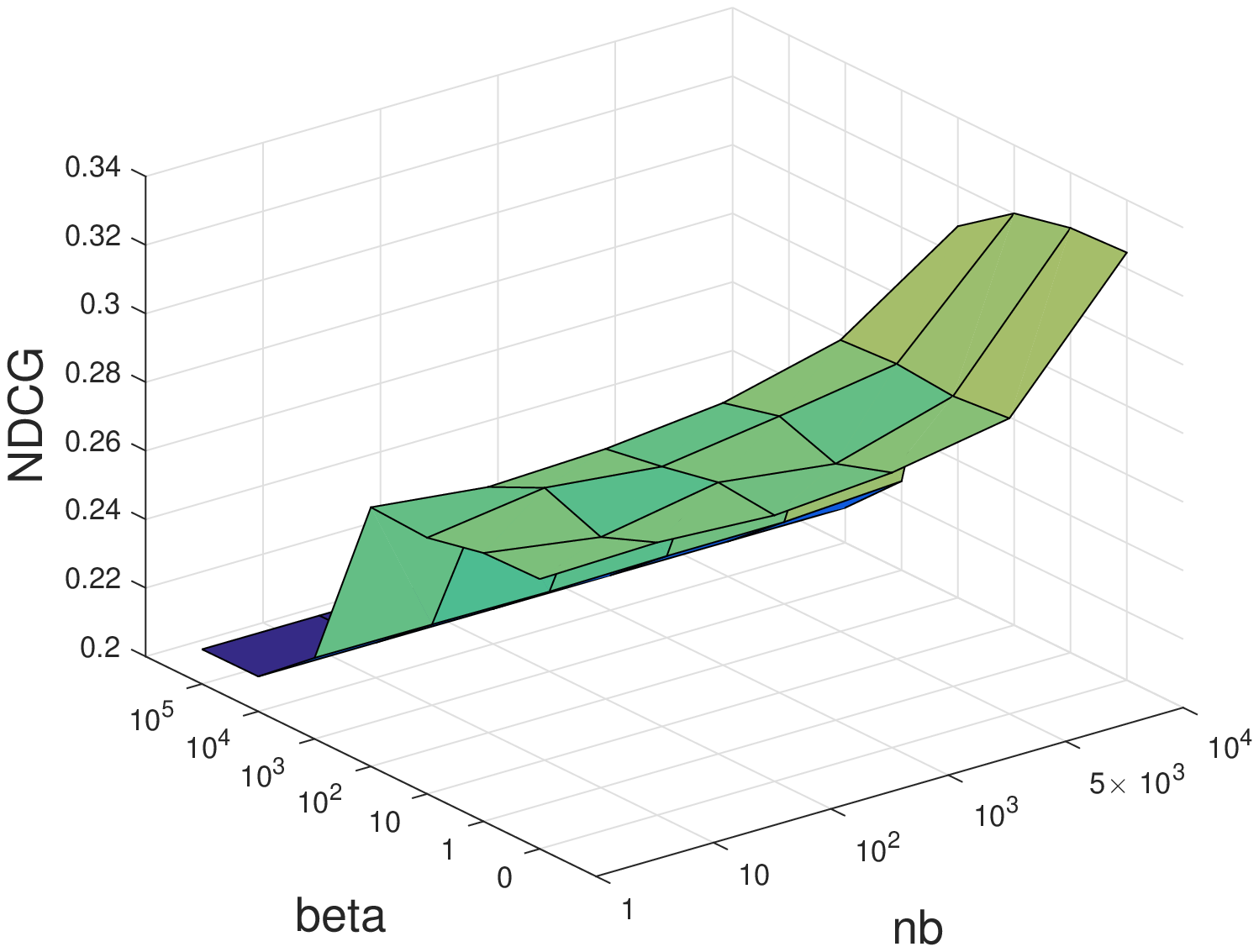}}
\subfigure[(d) NUS@ GSDH\_P]{\includegraphics[trim={1em 1em 1em 1em}, clip, width=0.16\textwidth]{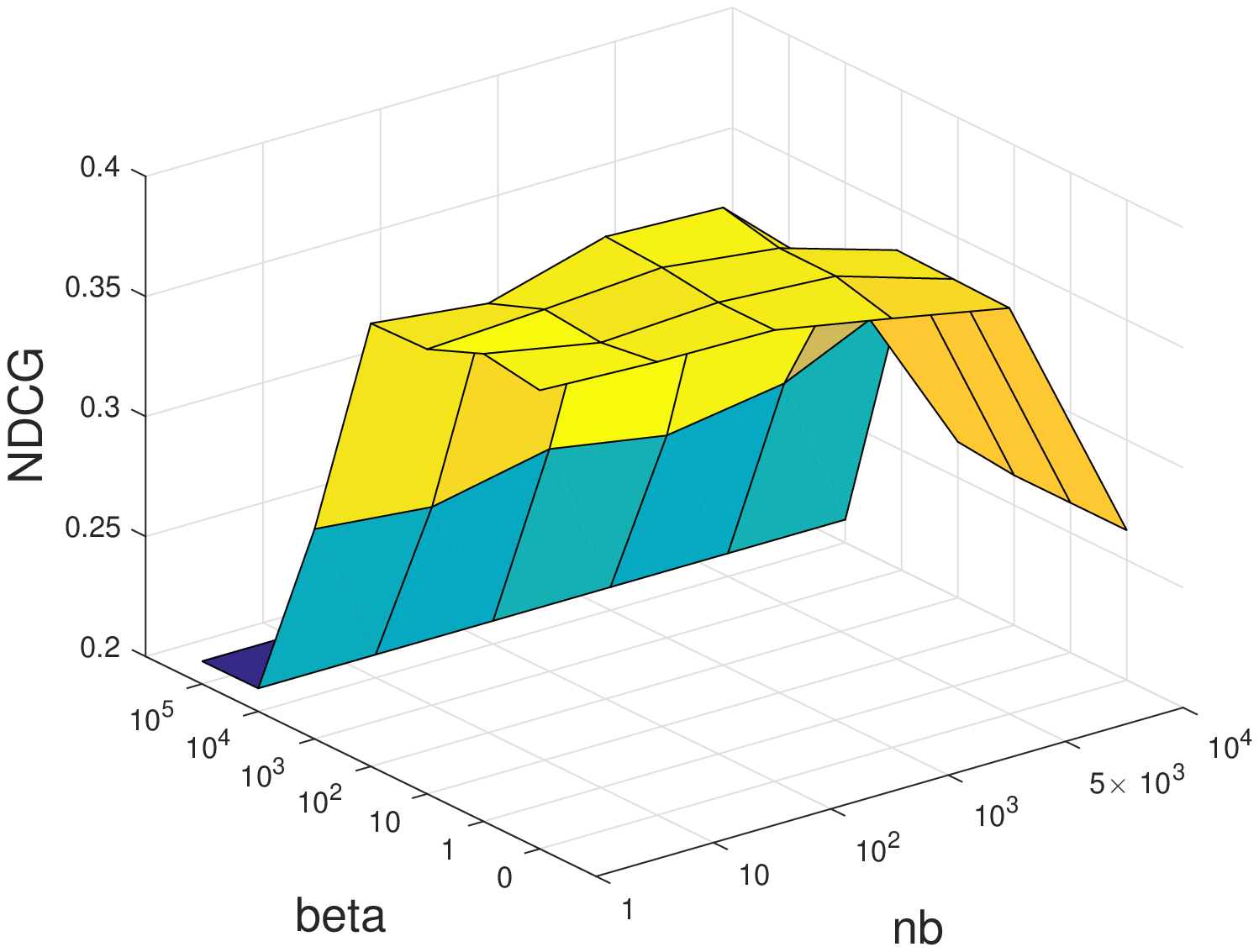}}
\subfigure[(e) COCO@ SDH\_P]{\includegraphics[trim={1em 1em 1em 1em}, clip, width=0.16\textwidth]{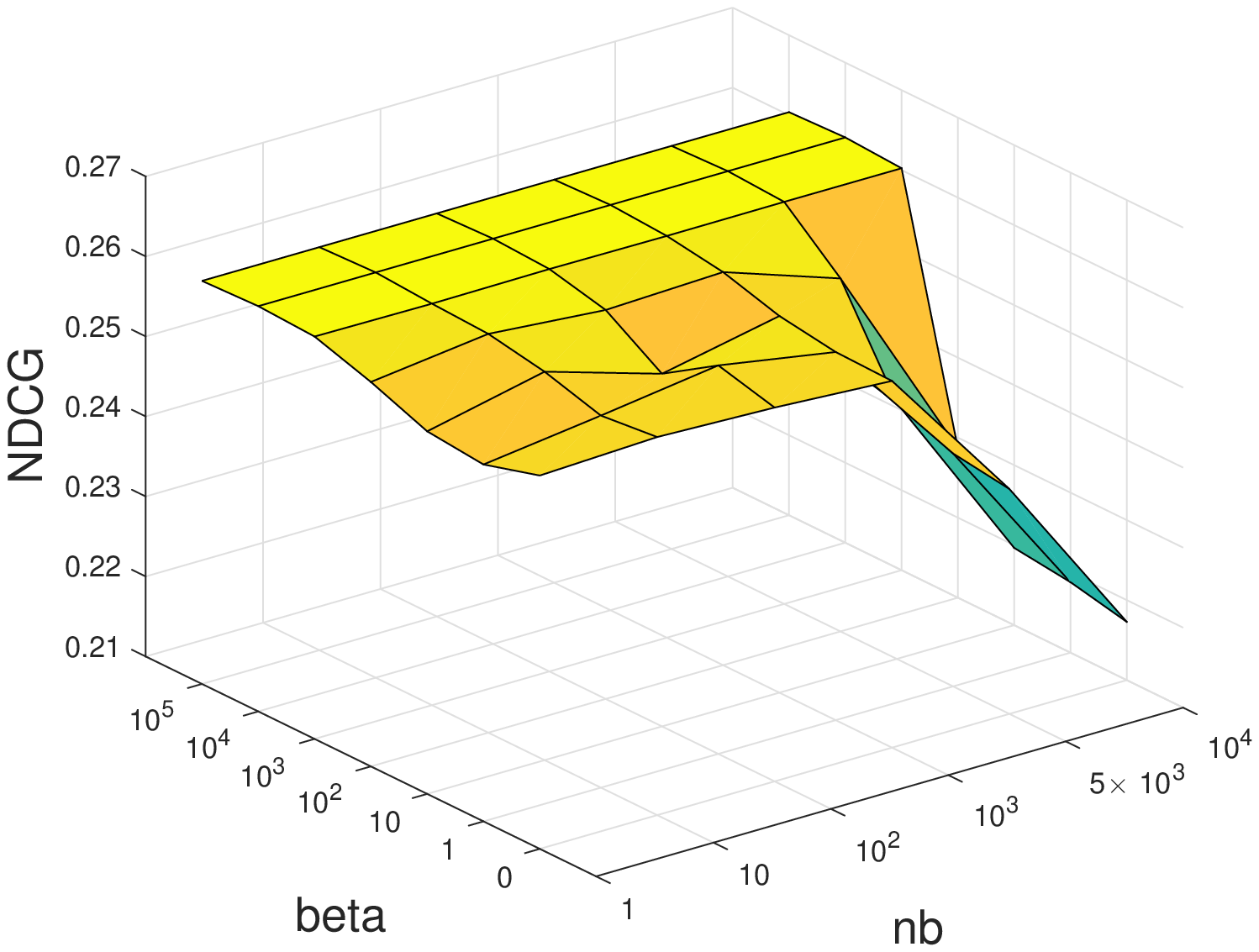}}
\subfigure[(f) COCO@ GSDH\_P]{\includegraphics[trim={1em 1em 1em 1em}, clip, width=0.16\textwidth]{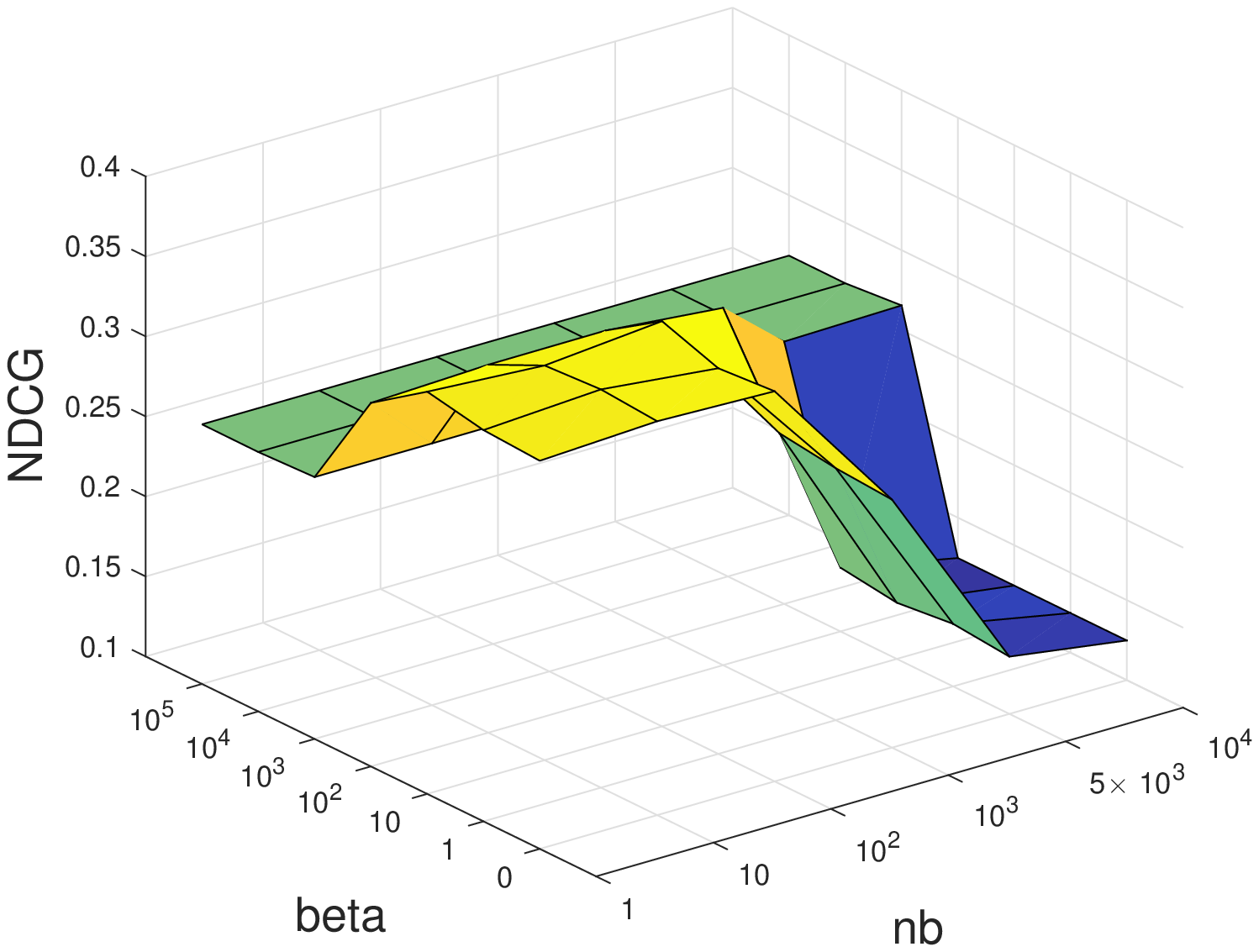}}
\vspace{-1em}
\caption{The influence of parameters \(n_{b}\) and \(\beta\) for SDH\_P and GSDH\_P on CIFAR-10, NUSWIDE and COCO databases.}
\label{fignpara}
\vspace{-1.5em}
\end{figure*}

\vspace{-1em}
\section{Conclusion}
\label{conclusion}
In this paper, we propose a novel, scalable and general optimization method to directly solve the non-convex and non-smooth problems in term of hash functions. We firstly solve a quartic problem in a least-squares model that utilizes two identical code matrices produced by hash functions to approximate a pairwise label matrix, by reformulating the quartic problem in term of hash functions into a quadratic problem, and then linearize it by introducing a linear regression model. Additionally, we find that gradually learning each batch of binary codes is beneficial to the convergence of learning process. Based on this finding, we propose a symmetric discrete hashing algorithm to gradually update each batch of the discrete matrix, and a greedy symmetric discrete hashing algorithm to greedily update each bit of batch discrete matrices. Finally, we extend the proposed greedy symmetric discrete hashing algorithm to handle other optimization problems.  Extensive experiments on single-label and multi-label databases demonstrate the effectiveness and efficiency of the proposed method. In this paper we only focus on solving the problems in pairwise based hashing, in the future, it is worth extending the proposed mechanism to solve the problems in multi-wise based hashing, whose objective is also highly non-differential, non-convex and more difficult to directly solve. 

\vspace{-0.5em}
{
\bibliographystyle{IEEETran}
\bibliography{ref}
}

\end{document}